\documentclass[dvips,preprint]{imsart}%
\usepackage{amsfonts}
\usepackage{amsmath}
\usepackage{epsfig}%
\usepackage{amssymb}%
\usepackage{graphicx}
\RequirePackage[OT1]{fontenc}
\RequirePackage{amsthm,amsmath}
\RequirePackage[numbers]{natbib}
\RequirePackage[colorlinks,citecolor=blue,urlcolor=blue]{hyperref}
\startlocaldefs
\numberwithin{equation}{section}
\theoremstyle{plain}
\newtheorem{theorem}{Theorem}

\newtheorem{lemma}{Lemma}
\newtheorem{proposition}{Proposition}
\newtheorem{remark}{Remark}

\endlocaldefs
\begin{document}
\begin{frontmatter}
\title{A CLT on the SNR of Diagonally Loaded MVDR Filters\thanksref{T1}}
\runtitle{A CLT on the SNR of the DL beamformer}
\thankstext{T1}{ $^\copyright$
2012 IEEE. Personal use of this material is permitted. However, permission to use this material for any other purposes must be obtained from the IEEE by sending a request to pubs-permissions@ieee.org. This paper was published at IEEE Transactions on Signal Processing, Septembe 2012. This version incorporates some corrections with respect to the published one.
This work was partially funded by the Spanish Ministry of Science and Innovation under project TEC2011-29006-C03-01
and Catalan Government under grant 2009SGR1046.
The material in this paper was presented in part at the
36th International Conference on Acoustics, Speech, and Signal Processing in Prague (Czech Republic), May 22-27, 2011.}
\begin{aug}
\author{\fnms{Francisco} \snm{Rubio}\thanksref{t1,m1}\ead
[label=e1]{franciscoerubio@gmail.com}},
\author{\fnms{Xavier} \snm{Mestre}\thanksref{m2}\ead
[label=e2]{xavier.mestre@cttc.cat}}
\affiliation{BB}
\and\author{\fnms{Walid} \snm{Hachem}\thanksref{m3}\ead
[label=e3]{walid.hachem@telecom-paristech.fr}}
\affiliation{CC}
\thankstext{t1}%
{The work was initiated while the first author was with the Centre Tecnològic de Telecomunicacions de Catalunya (CTTC), 08860 Castelldefels (Barcelona),
and was completed while he was with the Hong Kong University of Science and Technology.}
\runauthor{Rubio, Mestre and Hachem}
\address{\thanksmark{m1}Genetic Finance Ltd.\\
Unit 2107, Harbour Centre\\
25 Harbour Road, Hong Kong\\
\printead{e1}\\
\phantom{E-mail:\ }}
\address{\thanksmark{m2}%
Centre Tecnol\`ogic de Telecomunicacions de Catalunya\\
Av. de Carl Friedrich Gauss, 7\\
08860 Castelldefels (Barcelona), Spain\\
\printead{e2}\\
\phantom{E-mail:\ }}
\address{\thanksmark{m3}%
T\' el\' ecom ParisTech and Centre National de la Recherche Scientifique\\
46, rue Barrault\\
75013 Paris, France\\
\printead{e3}\\
\phantom{E-mail:\ }}
\end{aug}
\begin{abstract}
This paper studies the fluctuations of the signal-to-noise ratio (SNR) of minimum variance distorsionless response (MVDR) filters implementing diagonal loading in the estimation of the covariance matrix. Previous results in the signal processing literature are generalized and extended by considering both spatially as well as temporarily correlated samples. Specifically, a central limit theorem (CLT) is established for the fluctuations of the SNR of the diagonally loaded MVDR filter, under both supervised and unsupervised training settings in adaptive filtering applications. Our second-order analysis is based on the Nash-Poincar\' e inequality and the integration by parts formula for Gaussian functionals, as well as classical tools from statistical asymptotic theory. Numerical evaluations validating the accuracy of the CLT confirm the asymptotic Gaussianity of the fluctuations of the SNR of the MVDR filter.
\end{abstract}
\end{frontmatter}%

\section{Introduction}

The minimum variance distorsionless response (MVDR) filter is a prominent
instance of multivariate filtering structure in statistical signal processing.
Regarded as Capon beamformer, the MVDR spatial filter is widely utilized in
sensor array signal processing applications, such as the estimation of the
waveform and/or power of a given signal of interest (SOI) \cite{VT02B,SM05B}.
The theoretically optimal Capon/MVDR spatial filter is constructed based on a
covariance matrix that is unknown in practice, and so any filter
implementation must rely on sample estimates computed from the array
observations available. Sample covariance estimators are well-known to be
prohibitively inaccurate for sample volumes of small size, relatively high
dimension. Indeed, a vast body of contributions in the literature of array
processing and other fields of applied statistics has been devoted to remedies
for lifting the curse of dimensionality, such as those based on regularization
techniques and shrinkage estimation.

In this work, we are interested in the signal-to-noise ratio (SNR) at the
output of MVDR filter realizations using a diagonally loaded sample covariance
matrix (SCM). We focus on the SNR as a measure conventionally used to evaluate
the performance of a filter implementation. Due to its dependence on the
sample data matrix, the SNR is itself a random variable whose behavior highly
depends on the ratio between sample size and observation dimension. This ratio
is indeed of much practical relevance for characterizing the properties of the
filter performance. Motivated by this fact, a large-system performance
characterization was presented in \cite[Proposition 1]{ML06}, where the
authors provide a deterministic equivalent of the output SNR in the limiting
regime defined by both the number of samples and the observation dimension
growing large without bound at the same rate (see also \cite{ML05B}).

A first-order asymptotic analysis precludes us from gaining any insight on the
fluctuations of the SNR performance measure. Therefore, our focus in this work
is on a second-order analysis of the previous quantity. In the case of
Gaussian observations, when the maximum likelihood estimator of the population
covariance matrix is applied without diagonal loading, the normalized output
SNR is known in the array processing literature to follow a Beta distribution
\cite{RMB74}. In the general and more relevant case for practical
implementations considering the application of diagonal loading, the problem
of characterizing the distribution of the previous random variable remains
unsolved. Earlier attempts focused on the output response of the classical
diagonally loaded Capon/MVDR beamformer, by approximating its probability
density function via the truncation of a matrix power series \cite{F00} (see
also introductory exposition therein for details on previous related work),
and for the particular cases of zero- and single-source scenarios \cite{RE05},
as well as a two-source scenario \cite{N07P}.

In this paper, we generalize previous studies by considering both the use of
diagonal loading as well as general spatio-temporally correlated observations.
Specifically, we prove the asymptotic Gaussianity of the sample performance
measure by establishing a central limit theorem (CLT) on the output SNR of a
diagonally loaded MVDR filter implementation. To that effect, we resort to a
set of techniques for Gaussian random matrices, namely the Nash-Poincare
inequality as well as the integration by parts formula for Gaussian
functionals. These tools were originally proposed in \cite{HKLNP08} for the
study of the asymptotic distribution of the mutual information of correlated
MIMO Rayleigh channels.\ More recently, they have also been applied, for
instance, to obtain asymptotic non-Gaussian approximations of the distribution
of the SNR of the linear minimum mean-square error (LMMSE) receiver
\cite{KKHN09}, as well as to derive the input covariance matrix maximizing the
capacity of correlated MIMO Rician channels \cite{DHLLN10}.

Our framework relies on a limiting regime defined as both dimensions of the
data matrix going to infinity at the same rate. Indeed, in real-life array
processing applications, both the number of samples and the dimension of the
array are comparable in magnitude, and so a limiting regime allowing for both
sample size and dimension growing large with a fixed, non-zero ratio between
them is of more practical relevance. We will consider both supervised and
unsupervised training methods in statistical signal and sensor array
processing applications (see, e.g., \cite{ASAG98,ASA00}). In the former,
access to SOI-free samples of the interference-plus-noise process is granted
for covariance matrix estimation (e.g., clutter statistics in space-time
adaptive processing applications to radar), whereas only SOI-contaminated
samples are available for inference in the latter.

The structure of the rest of the paper after the previous exposition of the
research motivation is as follows. Upon concluding this section by introducing
the notation that will be used throughout the paper, Section
\ref{sectionMVDR_filtering} briefly presents the problem of multivariate
minimum variance filtering; the typical implementation based on a diagonally
loaded sample covariance matrix (SCM) is introduced along with the definition
of SNR as performance measure of relevance. In Section \ref{sectionMain} we
establish the CLT for the fluctuations of the SNR performance of both
supervised and unsupervised training methods. In Section \ref{sectionGauss},
we introduce the main mathematical tools for our analysis and state some
preliminary results serving as preparation for the proof of the CLT. Our
result on the asymptotic Gaussianity of SNR measures is numerically validated
in Section \ref{sectionSimulations},\ before concluding the paper with Section
\ref{sectionConclusions}. The technical details of the proof of the CLT in
Section \ref{sectionGauss} are postponed to the appendices.

\textbf{Notation}. In this paper, we use the following notations. All vectors
are defined as column vectors and designated with bold lower case; all
matrices are given in bold upper case; for both vectors and matrices a
subscript will be added to emphasize dependence on dimension, though it will
be occasionally dropped for the sake of clarity of presentation; $\left[
\cdot\right]  _{ij}$ will be used with matrices to extract the entry in the $i
$th row of the $j$th column, $\left[  \cdot\right]  _{j}$ will be used for the
$j$th entry of a vector or the nonzero elements of a diagonal matrix; $\left(
\cdot\right)  ^{T}$ denotes transpose; $\left(  \cdot\right)  ^{\ast}$ denotes
Hermitian (i.e. complex conjugate transpose); $\mathbf{I}_{M}$ denotes the
$M\times M$ identity matrix; $\operatorname*{tr}\left[  \cdot\right]  $
denotes the matrix trace operator; $\mathbb{R}$ and $\mathbb{C}$ denote the
real and complex fields of dimension specified by a superscript;
$\mathbb{R}^{\mathbb{+}}$ denotes the set of positive real numbers;
$\mathbb{P}\left(  \cdot\right)  $ denotes the probability of a random event,
$\mathbb{E}\left[  \cdot\right]  $ denotes the expectation operator, and
$\operatorname*{var}\left(  \cdot\right)  $ and $\operatorname*{cov}\left(
\cdot,\cdot\right)  $ denote, respectively, variance and covariance; $K,K_{p}$
denote constant values not depending on any relevant quantity, apart from the
latter on a parameter $p$; $\left\vert \cdot\right\vert $ denotes absolute
value; for any two functions $f_{N},g_{N}$ depending on $N$, $f_{N}%
=\mathcal{O}\left(  g_{N}\right)  $ will denote the fact that $\left\vert
f_{N}\right\vert \leq K\left\vert g_{N}\right\vert $, for sufficiently large
$N$, and $f_{N}=o_{p}\left(  1\right)  $ will denote convergence in
probability to zero of $f_{N}$; $\left\Vert \cdot\right\Vert $ denotes the
Euclidean norm for vectors and the induced norm for matrices (i.e. spectral or
strong norm), whereas $\left\Vert \cdot\right\Vert _{F}$ and $\left\Vert
\cdot\right\Vert _{\operatorname*{tr}}$ denote the Frobenius norm and trace
(or nuclear) norm, respectively, i.e., for a matrix $\mathbf{A}\in
\mathbb{C}^{M\times M}$ with eigenvalues $\lambda_{m},m=1,\ldots,M$ and
spectral radius $\rho\left(  \mathbf{A}\right)  =\max_{1\leq m\leq M}\left(
\left\vert \lambda_{m}\right\vert \right)  $, $\left\Vert \mathbf{A}%
\right\Vert =\left(  \rho\left(  \mathbf{A}^{\ast}\mathbf{A}\right)  \right)
^{1/2}$, $\left\Vert \mathbf{A}\right\Vert _{F}=\left(  \operatorname*{tr}%
\left[  \mathbf{A}^{\ast}\mathbf{A}\right]  \right)  ^{1/2}$ and $\left\Vert
\mathbf{A}\right\Vert _{\operatorname*{tr}}=\operatorname*{tr}\left[  \left(
\mathbf{A}^{\ast}\mathbf{A}\right)  ^{1/2}\right]  $.

\section{MVDR filtering with diagonal loading%
\label{sectionMVDR_filtering}%
}

In this section, we introduce the signal model and briefly review the problem
of spatial or multivariate MVDR filtering motivating our research. Let
$\mathbf{Y}_{\beta,N}=\left[  \mathbf{y}_{\beta}\left(  1\right)
,\ldots,\mathbf{y}_{\beta}\left(  N\right)  \right]  $ be the data matrix with
sample observations in a statistical signal processing application, where the
parameter $\beta$ indicates presence ($\beta=1$) or not ($\beta=0$) of the SOI
in the observations, which are modeled as:%
\begin{equation}
\mathbf{y}_{\beta}\left(  n\right)  =\beta s\left(  n\right)  \mathbf{s}%
+\mathbf{n}\left(  n\right)  \in\mathbb{C}^{M},\quad1\leq n\leq N\label{DGP}%
\end{equation}
where $s\left(  n\right)  $ is the waveform process of a given SOI, the vector
$\mathbf{s}$ models the SOI signature, and $\mathbf{n}\left(  n\right)  $
represents the contribution from some colored interference and the
cross-sectionally uncorrelated background noise, which we model jointly as a
zero-mean Gaussian process with covariance matrix $\mathbf{R}_{0,M}$. Signal
and interference-plus-noise processes are assumed to be independent.
Additionally, without loss of generality we will assume that the SOI power is
$1$, and also that $\left\Vert \mathbf{s}\right\Vert =1$. In particular, we
consider applications relying on \textit{supervised training}, where
$\mathbf{Y}_{\beta,N}=\mathbf{Y}_{0,N}$ contains SOI-free samples of the
interference-plus-noise process, or \textit{unsupervised training}, where the
training samples in $\mathbf{Y}_{\beta,N}=\mathbf{Y}_{1,N}$ are contaminated
by the SOI. Notice that each observation $\mathbf{y}_{\beta}\left(  n\right)
$ might be modeling the matched filter output sufficient statistic for the
received unknown symbols $s\left(  n\right)  $ at a multiuser detector in a
communications application, where $\mathbf{s}$ is the effective user
signature; or an array processor, where $\mathbf{s}$ contains the angular
frequency information (steering vector) related to the intended source,
represented by $s\left(  n\right)  $.

In order to allow for a more general signal modeling context, we consider the
case in which the vector observations are not only spatially or
cross-sectionally correlated but also present a certain correlation in the
time domain. This is typically the case in array processing applications where
the sources exhibit nonzero correlation between delayed samples \cite{VSO95},
as well as generally for wireless communication signals that are transmitted
over a dispersive radio channel. In this work, we consider spatio-temporal
processes with separable covariance structure, also regarded as having
Kronecker product structure, and thoroughly studied in the literature on
multiple-input multiple-output wireless communication channels \cite{KSPMF02},
and sensor array and multichannel processing \cite{S95b}. In particular, the
spatial covariance matrix will be denoted by $\mathbf{R}_{\beta,M}$, and the
time correlation pattern will be modeled by a nonnegative matrix denoted by
$\mathbf{T}_{N}$, so that the column vectors of $\mathbf{Y}_{N}$ are
correlated (in the time domain) but the correlation pattern is identical for
all rows. Notice that the spatial covariance matrix $\mathbf{R}_{\beta,M}$ is
intrinsically different depending on the type of training, i.e.,
$\mathbf{R}_{\beta,M}\equiv\mathbf{R}_{0,M}$ for supervised training, and
$\mathbf{R}_{\beta,M}\equiv\mathbf{R}_{1,M}=\mathbf{ss}^{\ast}+\mathbf{R}%
_{0,M}$ for unsupervised training. As an illustrative example, consider the
following first-order vector autoregressive process: $\mathbf{y}_{\beta
}\left(  n\right)  =\psi\mathbf{y}_{\beta}\left(  n-1\right)  +\mathbf{R}%
_{\beta,M}^{1/2}\mathbf{v}\left(  n\right)  $, where $\psi$ is a real-valued
constant and $\mathbf{v}\left(  n\right)  $ is a white Gaussian noise process
with zero mean and identity covariance matrix, and $\mathbf{R}_{\beta,M}%
^{1/2}$ is a square-root of a positive matrix $\mathbf{R}_{\beta,M} $. In
particular, the previous so-called VAR(1) model has covariance matrix with
separable (Kronecker product) structure given by $\operatorname*{cov}\left(
\left[  \mathbf{y}_{\beta}\left(  n\right)  \right]  _{i},\left[
\mathbf{y}_{\beta}\left(  n+\tau\right)  \right]  _{j}\right)  =\left[
\mathbf{R}_{\beta,M}\right]  _{ij}\psi/\left(  1-\psi^{2}\right)  ^{\left\vert
\tau\right\vert }$.

Motivated by typical applications in sensor array signal processing, in this
paper we concentrate on the problem of linearly filtering the observed samples
with a Capon/MVDR beamformer to estimate the SOI waveform assuming that the
SOI signature is known. We notice that a related problem that is not handled
here but can also be fitted into our framework is that of estimating the SOI
power \cite{SM05B}. Customarily, the problem of optimizing the coefficients of
the Capon/MVDR spatial filter is formulated in terms of the spatial covariance
matrix as:%
\[
\mathbf{w}_{\beta,\mathsf{MVDR}}=\arg\min_{\mathbf{w}\in\mathbb{C}%
^{M}:\mathbf{w}^{\ast}\mathbf{s}=1}\mathbf{w}^{\ast}\mathbf{R}_{\beta
,M}\mathbf{w}%
\]
with explicit solution being given by%
\begin{equation}
\mathbf{w}_{\beta,\mathsf{MVDR}}=\frac{\mathbf{R}_{\beta,M}^{-1}\mathbf{s}%
}{\mathbf{s}^{\ast}\mathbf{R}_{\beta,M}^{-1}\mathbf{s}}\text{.}%
\label{OFS_beta}%
\end{equation}
Under the above conventional assumptions, the two previous covariance matrices
differ by the rank-one matrix term $\mathbf{ss}^{\ast}$, and so it is easy to
see that the optimal solutions with $\mathbf{R}_{0,M}$ and $\mathbf{R}_{1,M}$
are equivalent, i.e., $\mathbf{w}_{0,\mathsf{MVDR}}=\mathbf{w}%
_{1,\mathsf{MVDR}}$. Conventionally, the evaluation of the performance of the
filter is based on the SNR measure, which is defined as%
\begin{equation}
\mathsf{SNR}\left(  \mathbf{w}\right)  =\frac{\left\vert \mathbf{w}^{\ast
}\mathbf{s}\right\vert ^{2}}{\mathbf{w}^{\ast}\mathbf{R}_{0,M}\mathbf{w}%
}\text{.}\label{SNR}%
\end{equation}
In particular, we have%
\footnote
{It is not difficult to see that the maximum SNR values for supervised and unsupervised training theoretically coincide. In practice, however, the actual performance of an unsupervised training method would be diminished by inaccuracies about the knowledge of the precise SOI signature, and therefore a supervised training method is preferred in this sense.}
$\mathsf{SNR}\left(  \mathbf{w}_{0,\mathsf{MVDR}}\right)  =\mathsf{SNR}\left(
\mathbf{w}_{1,\mathsf{MVDR}}\right)  =\mathbf{s}^{\ast}\mathbf{R}_{0,M}%
^{-1}\mathbf{s}\equiv\mathsf{SNR}_{\mathsf{opt}}$.

In practice, the covariance matrix is unknown and so any implementation of the
filter must rely on estimates built upon a set of training samples. The
standard SCM estimator is usually improved by means of, for instance,
regularization or shrinkage. In particular, we consider covariance matrix
estimators of the type%
\begin{equation}
\mathbf{\hat{R}}_{\beta,M}=\frac{1}{N}\mathbf{Y}_{\beta,N}\mathbf{Y}_{\beta
,N}^{\ast}+\alpha\mathbf{I}_{M}\label{SCM}%
\end{equation}
where $\alpha>0$ is a constant scalar that in the array processing literature
is referred to as diagonal loading factor, and is also known in the statistics
literature as shrinkage intensity parameter for the type of James-Stein
shrinkage covariance matrix estimators. In brief, the purpose of the
regularization term $\alpha\mathbf{I}$ is to improve the condition number of
an \textit{a priory} possibly unstable estimator of the covariance matrix of
the array observations. This is particularly the case for the SCM in
situations where $N$ is not considerably larger than $M$. Indeed, notice that
the SCM might not even be invertible, as it happens in the case $M>N$.
Well-conditioned covariance matrix estimators can be expected to improve the
filter performance as measured by the realized SNR defined in (\ref{SNR}). In
this work, we assume that the parameter $\alpha$ is given and fixed. For
sensible choices of the regularization or diagonal loading parameter $\alpha$,
we refer the reader to, e.g., \cite{VT02B,LS05B}.

We now handle the situation in which a covariance matrix estimator of the type
of (\ref{SCM}) is used in order to implement the sample version of the
theoretical MVDR filter, which will be denoted in the sequel by $\mathbf{\hat
{w}}_{\beta,\mathsf{MVDR}}$, $\beta=0,1$. Then, using $\mathbf{w}%
=\mathbf{\hat{w}}_{\beta,\mathsf{MVDR}}$ in (\ref{SNR}), we obtain,
respectively,%
\begin{equation}
\mathsf{SNR}\left(  \mathbf{\hat{w}}_{0,\mathsf{MVDR}}\right)  =\frac{\left(
\mathbf{s}^{\ast}\mathbf{\hat{R}}_{0,M}^{-1}\mathbf{s}\right)  ^{2}%
}{\mathbf{s}^{\ast}\mathbf{\hat{R}}_{0,M}^{-1}\mathbf{R}_{0,M}\mathbf{\hat{R}%
}_{0,M}^{-1}\mathbf{s}}\label{sSNRs}%
\end{equation}
and%
\begin{equation}
\mathsf{SNR}\left(  \mathbf{\hat{w}}_{1,\mathsf{MVDR}}\right)  =\left(
\frac{\mathbf{s}^{\ast}\mathbf{\hat{R}}_{1,M}^{-1}\mathbf{R}_{1,M}%
\mathbf{\hat{R}}_{1,M}^{-1}\mathbf{s}}{\left(  \mathbf{s}^{\ast}%
\mathbf{\hat{R}}_{1,M}^{-1}\mathbf{s}\right)  ^{2}}-1\right)  ^{-1}%
\text{.}\label{sSNRu}%
\end{equation}
Equations (\ref{sSNRs}) and (\ref{sSNRu}) are obtained by directly replacing
in (\ref{SNR}) the optimal MVDR filter solution in (\ref{OFS_beta}) for,
respectively, the supervised ($\beta=0$) and unsupervised ($\beta=1$) cases.
Notice that, while the expression in (\ref{sSNRs}) follows straightforwardly,
in order to get (\ref{sSNRu}) it is enough to apply the matrix inversion lemma
using the fact that $\mathbf{R}_{1,M}=\mathbf{ss}^{\ast}+\mathbf{R}_{0,M}$.

In effect, due to the dependence on the random data matrix $\mathbf{Y}%
_{\beta,N}$, the quantities (\ref{sSNRs}) and (\ref{sSNRu}) are random
variables themselves whose distribution specify the fluctuations of the SNR
performance at the filter output. Consequently, in order to understand the
behavior of the output SNR performance, it is of much practical interest to
investigate the distribution of the random variables (\ref{sSNRs}) and
(\ref{sSNRu}), and characterize their properties. Under the supervised
training setting, in the special case given by $\mathbf{\hat{R}}_{\beta,M}$
being the standard SCM estimator, i.e., $\mathbf{T}_{N}=\mathbf{I}_{N}$ and
$\alpha=0$, the distribution of the normalized output SNR, namely,%
\[
\frac{\mathsf{SNR}\left(  \mathbf{\hat{w}}_{0,\mathsf{MVDR}}\right)
}{\mathsf{SNR}_{\mathsf{opt}}}=\frac{\left(  \mathbf{s}^{\ast}\mathbf{\hat{R}%
}_{0,M}^{-1}\mathbf{s}\right)  ^{2}}{\mathbf{s}^{\ast}\mathbf{\hat{R}}%
_{0,M}^{-1}\mathbf{R}_{0,M}\mathbf{\hat{R}}_{0,M}^{-1}\mathbf{ss}^{\ast
}\mathbf{R}_{0,M}^{-1}\mathbf{s}}%
\]
is known to be distributed as \cite{RMB74}%
\[
\mathsf{SNR}\left(  \mathbf{\hat{w}}_{0,\mathsf{MVDR}}\right)  /\mathsf{SNR}%
_{\mathsf{opt}}\sim\operatorname*{Beta}\left(  N+2-M,M-1\right)  .
\]
In the general, more relevant case for practical implementations, where
arbitrary positive definite $\mathbf{T}_{N}$ and $\alpha$ are considered, the
problem of characterizing the distribution of the random variable
$\mathsf{SNR}\left(  \mathbf{\hat{w}}_{0,\mathsf{MVDR}}\right)  $ remains
unsolved. Likewise, so is the case for $\mathsf{SNR}\left(  \mathbf{\hat{w}%
}_{1,\mathsf{MVDR}}\right)  $.

In the next section, we provide a CLT on the realized SNR performance at the
output of a sample MVDR filter implementing diagonal loading and based on a
set of spatio-temporally correlated observations, for both supervised and
unsupervised training applications. We remark that in this paper we are
specifically concerned with the case $\alpha>0$. In fact, the case $\alpha=0$
has been seldom considered in the large random matrix literature, and would
require indeed specific tools different from those used here.

\section{CLT for the fluctuations of SNR performance measures%
\label{sectionMain}%
}

\subsection{Definitions and assumptions%
\label{sectionDefAss}%
}

We next summarize our research hypotheses and introduce some new definitions.
We first remark that, anticipating that the statistical properties of the
random matrices $\mathbf{Y}_{\beta,M}$ and $\mathbf{\hat{R}}_{\beta,M}$ for
both values of $\beta$ are equivalent for the purposes of our derivations, we
will drop the subscript $\beta$ in the sequel. Our analysis is based on the
following technical hypotheses:

\begin{description}
\item[\textbf{(As1)}] The observations are normally distributed with zero mean
and separable covariances $\mathbf{R}_{M}$ and $\mathbf{T}_{N}$ in the spatial
and time domain respectively.

\item[\textbf{(As2)}] The nonrandom matrices $\mathbf{R}_{M}$ and
$\mathbf{T}_{N}$ have eigenvalues bounded uniformly in, respectively, $M$ and
$N=N\left(  M\right)  $, from above, i.e., $\left\Vert \mathbf{R}\right\Vert
_{\sup}=\sup_{M\geq1}\left\Vert \mathbf{R}_{M}\right\Vert <+\infty$ and
$\left\Vert \mathbf{T}\right\Vert _{\sup}=\sup_{N\geq1}\left\Vert
\mathbf{T}_{N}\right\Vert <+\infty$, and from below (away from zero):
$\left\Vert \mathbf{R}\right\Vert _{\inf}=\inf_{M\geq1}\left\Vert
\mathbf{R}_{M}^{-1}\right\Vert ^{-1}>0$ and $\left\Vert \mathbf{T}\right\Vert
_{\inf}=\inf_{N\geq1}\left\Vert \mathbf{T}_{N}^{-1}\right\Vert ^{-1}>0$.

\item[\textbf{(As3)}] We will consider the limiting regime defined by both
dimensions $M$ and $N$ growing large without bound at the same rate, i.e.,
$N,M\rightarrow\infty$ such that ($c_{M}=M/N$):%
\[
0<c_{\inf}=\lim\inf c_{M}\leq c_{\sup}=\lim\sup c_{M}<\infty\text{.}%
\]

\end{description}

Let $\mathbf{X}_{M}$ be an $M\times N$ matrix whose elements $X_{ij}$, $1\leq
i\leq M$, $1\leq j\leq N$, are complex Gaussian random variables having i.i.d.
real and imaginary parts with mean zero and variance $1/2$, such that
$\mathbb{E}\left[  X_{ij}\right]  =\mathbb{E}\left[  X_{ij}^{2}\right]  =0$
and $\mathbb{E}\left[  \left\vert X_{ij}\right\vert ^{2}\right]  =1$. Under
the Gaussianity assumption, observe that we can write the data matrix in
Section \ref{sectionMVDR_filtering} as $\mathbf{Y}_{N}=\mathbf{R}_{M}%
^{1/2}\mathbf{X}_{M}\mathbf{T}_{N}^{1/2}$, where $\mathbf{R}_{M}^{1/2}$ and
$\mathbf{T}_{N}^{1/2}$ are the positive definite square-roots of
$\mathbf{R}_{M}$ and $\mathbf{T}_{N}$, respectively. Hence, the data matrix
$\mathbf{Y}_{N}$ is matrix-variate normal distributed, i.e., $\mathbf{Y}%
_{N}\sim\mathcal{CMN}_{M\times N}\left(  \mathbf{0}_{M\times N},\mathbf{R}%
_{M},\mathbf{T}_{N}\right)  $, or equivalently, $\operatorname*{vec}\left(
\mathbf{Y}_{N}\right)  \sim\mathcal{CN}_{MN}\left(  \mathbf{0}_{M}%
,\mathbf{R}_{M}\otimes\mathbf{T}_{N}\right)  $ \cite{GN00B}. Moreover, in the
case of an arbitrary positive definite matrix $\mathbf{T}_{N}$, we have that
$\mathbf{Y}_{N}\mathbf{Y}_{N}^{\ast}$ is a central quadratic form, such that
$\mathbb{E}\left[  \mathbf{Y}_{N}\mathbf{Y}_{N}^{\ast}\right]
=\operatorname*{Tr}\left[  \mathbf{T}_{N}\right]  \mathbf{R}_{M}$. Thus, in
particular, if $\mathbf{T}_{N}=\mathbf{I}_{N}$ then $\mathbf{Y}_{N}%
\mathbf{Y}_{N}^{\ast}$ is central Wishart distributed, and we have
$\mathbb{E}\left[  \mathbf{Y}_{N}\mathbf{Y}_{N}^{\ast}\right]  =N\mathbf{R}%
_{M}$ (see also, e.g., \cite[Chapter 2]{KR05}). We note that our
spatio-temporal covariance model represents a non-trivial generalization of
previous models, which is of interest for the signal processing and the
applied statistics community. For instance, the model in \cite{R96,ZLY11}
consisting of a data matrix $\mathbf{Y}_{M}=\mathbf{R}_{M}^{1/2}\mathbf{\Xi
}_{N}$, where $\mathbf{R}_{M}=\mathbb{E}\left[  \mathbf{Y}_{M}\mathbf{Y}%
_{M}^{\ast}\right]  $ and $\mathbf{\Xi}_{N}$ is a Gaussian matrix with
standardized entries (i.e., with mean zero and variance one), is clearly a
special case of our model.

We recall that the previous distributional assumption is fairly standard in
the array processing literature (e.g., \cite{RMB74,F00,RE05}, and
\cite{R96,ZLY11}). In particular, the Gaussianity assumption provides a means
to obtain valuable approximations of the system performance by analytically
characterizing the theoretical properties of otherwise intractable expressions
of practical interest. On the other hand, the assumption of centered
observations has minor impact, since observations can always be demeaned by
extracting the sample mean. In fact, for Gaussian sample observations, the
sample covariance matrix with and without estimation of the mean has the
Wishart structure described above (with one degree-of-freedom less in the case
of having to estimate the mean, which does not affect our asymptotic results).

Before proceeding any further, we also notice that, thanks to the isotropic
invariance to orthogonal transformations of Gaussian matrices, the two
correlation matrices $\mathbf{R}_{M}$ and $\mathbf{T}_{N}$ can be assumed to
be diagonal without loss of generality. More specifically, using the fact that
the distribution of a Gaussian matrix is unaffected by unitary
transformations, it is easy to see that we can always write the SNR in
(\ref{sSNRs}) and (\ref{sSNRu}) in terms of a unit-norm deterministic vector,
a Gaussian matrix with standardized entries, and diagonal spatial and temporal
covariance matrices. Such a parsimonious representation is more convenient for
proving our statistical results, and is therefore preferred.

We next introduce some notation that will be useful throughout the rest of the
paper. Let us first introduce the vector $\mathbf{u}_{M}=\mathbf{R}_{M}%
^{-1/2}\mathbf{s}$ and the matrix $\mathbf{Q}_{M}=\left(  \frac{1}%
{N}\mathbf{X}_{N}\mathbf{T}_{N}\mathbf{X}_{N}^{\ast}+\alpha\mathbf{R}_{M}%
^{-1}\right)  ^{-1}$, where $\alpha>0$. Moreover, we define
\begin{equation}
a_{M}=\mathbf{u}_{M}^{\ast}\mathbf{Q}_{M}\mathbf{u}_{M}\quad b_{M}%
=\mathbf{u}_{M}^{\ast}\mathbf{Q}_{M}^{2}\mathbf{u}_{M},\label{eq_defab}%
\end{equation}
along with
\begin{equation}
\bar{a}_{M}=\mathbf{u}_{M}^{\ast}\mathbf{E}_{M}\mathbf{u}_{M}\quad\bar{b}%
_{M}=\left(  1-\gamma_{M}\tilde{\gamma}_{M}\right)  ^{-1}\mathbf{u}_{M}^{\ast
}\mathbf{E}_{M}^{2}\mathbf{u}_{M}\text{.}\label{eq_def_abbar}%
\end{equation}
where $\gamma=\gamma_{M}=\frac{1}{N}\operatorname*{tr}\left[  \mathbf{E}%
_{M}^{2}\right]  $ and $\tilde{\gamma}=\tilde{\gamma}_{M}=\frac{1}%
{N}\operatorname*{tr}\left[  \mathbf{\tilde{E}}_{N}^{2}\right]  $, with
\begin{gather*}
\mathbf{E}_{M}=\mathbf{R}_{M}\left(  \tilde{\delta}_{M}\mathbf{R}_{M}%
+\alpha\mathbf{I}_{M}\right)  ^{-1}\\
\mathbf{\tilde{E}}_{N}=\mathbf{T}_{N}\left(  \mathbf{I}_{N}+\delta
_{M}\mathbf{T}_{N}\right)  ^{-1}%
\end{gather*}
and $\left\{  \tilde{\delta}_{M},\delta_{M}\right\}  $ being the unique
positive solution to the following system of equations:
\begin{equation}
\left\{
\begin{array}
[c]{l}%
\tilde{\delta}_{M}=\frac{1}{N}\operatorname*{tr}\left[  \mathbf{T}_{N}\left(
\mathbf{I}_{N}+\delta_{M}\mathbf{T}_{N}\right)  ^{-1}\right] \\
\delta_{M}=\frac{1}{N}\operatorname*{tr}\left[  \mathbf{R}_{M}\left(
\tilde{\delta}_{M}\mathbf{R}_{M}+\alpha\mathbf{I}_{M}\right)  ^{-1}\right]
\text{.}%
\end{array}
\right. \label{systemEq}%
\end{equation}
The existence and uniqueness of the solution to (\ref{systemEq}) follow by
similar arguments as those in the proof of Proposition 1 in \cite{HKLNP08}.
Additionally, notice that $\mathbf{E}_{M}$ and $\mathbf{\tilde{E}}_{N}$ are
positive definite matrices. Before concluding, a final remark is in order.
Under Assumption \textbf{(As2)}, all previously defined elements are well
defined for all $M$ in the sense of the Euclidean norm for vectors or induced
norm for matrices (see uniform bounds provided at the end of Appendix
\ref{Appendix1}, which will be useful for the derivation of our asymptotic results).

\subsection{First-order approximations}

The following proposition provides asymptotic approximations for the expected
values of the random variables $a_{M}$ and $b_{M}$. The result follows readily
from Proposition \ref{theoremEV1} and Proposition \ref{theoremEV2} in Section
\ref{sectionGauss}.%

\begin{lemma}%
\label{propositionBias}%
With the definitions and under the assumptions above, the following
expectations hold:%
\begin{align*}
\mathbb{E}\left[  a_{M}\right]   &  =\bar{a}_{M}+\mathcal{O}\left(
N^{-3/2}\right) \\
\mathbb{E}\left[  b_{M}\right]   &  =\bar{b}_{M}+\mathcal{O}\left(
N^{-3/2}\right)  \text{.}%
\end{align*}%
\end{lemma}%

Based on the previous approximation rules, we will consider the following two
first-order estimates of the SNR under the supervised and the unsupervised
training settings, namely, $\overline{\mathsf{SNR}\left(  \mathbf{\hat{w}%
}_{0,\mathsf{MVDR}}\right)  }=\bar{a}_{M}^{2}/\bar{b}_{M}$, and $\overline
{\mathsf{SNR}\left(  \mathbf{\hat{w}}_{1,\mathsf{MVDR}}\right)  }=\left(
\bar{b}_{M}/\bar{a}_{M}^{2}-1\right)  ^{-1}$, respectively.

\subsection{Second-order analysis}

The following two theorems establish the asymptotic Gaussianity of the
fluctuations of the SNR performance measures (\ref{sSNRs}) and (\ref{sSNRu}).
Before stating the results, we introduce the following quantity, which is
shown to be positive in Section \ref{sectionCLT}:%
\begin{multline*}
\mathcal{V}_{M}=\tilde{\gamma}^{2}\frac{1}{N}\operatorname*{tr}\left[
\mathbf{E}_{M}^{4}\right]  +\gamma^{2}\frac{1}{N}\operatorname*{tr}\left[
\mathbf{\tilde{E}}_{N}^{4}\right] \\
+4\tilde{\gamma}\left(  1-\gamma\tilde{\gamma}\right)  \mathcal{S}%
_{M}+4\left(  \tilde{\gamma}^{2}\frac{1}{N}\operatorname*{tr}\left[
\mathbf{E}_{M}^{3}\right]  -\gamma\frac{1}{N}\operatorname*{tr}\left[
\mathbf{\tilde{E}}_{N}^{3}\right]  \right)  \mathcal{T}_{M}\\
+\frac{2}{\left(  1-\gamma\tilde{\gamma}\right)  }\left(  \tilde{\gamma}%
^{3}\left(  \frac{1}{N}\operatorname*{tr}\left[  \mathbf{E}_{M}^{3}\right]
\right)  ^{2}-2\gamma\tilde{\gamma}\frac{1}{N}\operatorname*{tr}\left[
\mathbf{E}_{M}^{3}\right]  \frac{1}{N}\operatorname*{tr}\left[  \mathbf{\tilde
{E}}_{N}^{3}\right]  +\gamma^{3}\left(  \frac{1}{N}\operatorname*{tr}\left[
\mathbf{\tilde{E}}_{N}^{3}\right]  \right)  ^{2}\right)
\end{multline*}
where we have defined%
\begin{gather*}
\mathcal{S}_{M}=\left(  \frac{\mathbf{u}_{M}^{\ast}\mathbf{E}_{M}%
^{2}\mathbf{u}_{M}}{\mathbf{u}_{M}^{\ast}\mathbf{E}_{M}\mathbf{u}_{M}}\right)
^{2}-2\frac{\mathbf{u}_{M}^{\ast}\mathbf{E}_{M}^{3}\mathbf{u}_{M}}%
{\mathbf{u}_{M}^{\ast}\mathbf{E}_{M}\mathbf{u}_{M}}+\frac{1}{2}\left[
\frac{\mathbf{u}_{M}^{\ast}\mathbf{E}_{M}^{4}\mathbf{u}_{M}}{\mathbf{u}%
_{M}^{\ast}\mathbf{E}_{M}^{2}\mathbf{u}_{M}}+\left(  \frac{\mathbf{u}%
_{M}^{\ast}\mathbf{E}_{M}^{3}\mathbf{u}_{M}}{\mathbf{u}_{M}^{\ast}%
\mathbf{E}_{M}^{2}\mathbf{u}_{M}}\right)  ^{2}\right] \\
\mathcal{T}_{M}=\frac{\mathbf{u}_{M}^{\ast}\mathbf{E}_{M}^{3}\mathbf{u}_{M}%
}{\mathbf{u}_{M}^{\ast}\mathbf{E}_{M}^{2}\mathbf{u}_{M}}-\frac{\mathbf{u}%
_{M}^{\ast}\mathbf{E}_{M}^{2}\mathbf{u}_{M}}{\mathbf{u}_{M}^{\ast}%
\mathbf{E}_{M}\mathbf{u}_{M}}\text{.}%
\end{gather*}
%

\begin{theorem}%
\label{theoremCLTs}%
(\textit{Supervised Training}) Under the definitions and assumptions in
Section \ref{sectionDefAss}, the following CLT holds:%
\[
\sigma_{s,M}^{-1}\sqrt{N}\left(  \mathsf{SNR}\left(  \mathbf{\hat{w}%
}_{0,\mathsf{MVDR}}\right)  -\overline{\mathsf{SNR}\left(  \mathbf{\hat{w}%
}_{0,\mathsf{MVDR}}\right)  }\right)  \overset{\mathcal{L}}{\rightarrow
}\mathcal{N}\left(  0,1\right)
\]
where%
\[
\sigma_{s,M}^{2}=\left(  \frac{\left(  \mathbf{u}_{M}^{\ast}\mathbf{\mathbf{E}%
}_{M}\mathbf{u}_{M}\right)  ^{2}}{\mathbf{u}_{M}^{\ast}\mathbf{E}_{M}%
^{2}\mathbf{u}_{M}}\right)  ^{2}\mathcal{V}_{M}\text{.}%
\]%
\end{theorem}%
%

\begin{theorem}%
\label{theoremCLTu}%
(\textit{Unsupervised Training}) Under the definitions and assumptions in
Section \ref{sectionDefAss}, the following CLT holds:%
\[
\sigma_{u,M}^{-1}\sqrt{N}\left(  \mathsf{SNR}\left(  \mathbf{\hat{w}%
}_{1,\mathsf{MVDR}}\right)  -\overline{\mathsf{SNR}\left(  \mathbf{\hat{w}%
}_{1,\mathsf{MVDR}}\right)  }\right)  \overset{\mathcal{L}}{\rightarrow
}\mathcal{N}\left(  0,1\right)
\]
where%
\[
\sigma_{u,M}^{2}=\left(  \frac{\left(  \mathbf{u}_{M}^{\ast}\mathbf{\mathbf{E}%
}_{M}\mathbf{u}_{M}\right)  ^{2}}{\mathbf{u}_{M}^{\ast}\mathbf{E}_{M}%
^{2}\mathbf{u}_{M}}\right)  ^{2}\left(  1-\frac{\left(  \mathbf{u}_{M}^{\ast
}\mathbf{\mathbf{E}}_{M}\mathbf{u}_{M}\right)  ^{2}}{\mathbf{u}_{M}^{\ast
}\mathbf{E}_{M}^{2}\mathbf{u}_{M}}\left(  1-\gamma_{M}\tilde{\gamma}%
_{M}\right)  \right)  ^{-4}\mathcal{V}_{M}\text{,}%
\]%
\end{theorem}%

The CLT's established in Theorems \ref{theoremCLTs} and \ref{theoremCLTu}
state the intricate but explicit dependence on the spatial and temporal
covariance matrices $\mathbf{R}_{M}$ and $\mathbf{T}_{N}$ of the mean and
variance of the realized SNR. In particular, notice that these two moments
univocally define the asymptotic Gaussian distributions derived above. Further
insights can be gained by a scenario-based analysis considering particular
choices of the covariances $\mathbf{R}_{M}$ and $\mathbf{T}_{N}$. Though
undoubtedly of practical relevance, such an analysis is outside of the scope
of this work, and left open for future research.

We remark that the previous analytical characterization of the asymptotic
distribution of the SNR for a given, fixed diagonal loading parameter, could
be used for selecting an improved parameter. Previous work by one of the
authors proposes a simple approach for fixing $\alpha$ by considering only a
first-order asymptotic analysis \cite{ML05B,ML06}. Potential approaches
exploiting the second-order asymptotic results provided here might be based on
determining the diagonal loading factor maximizing not only the expected value
of the realized SNR, but a linear combination of the mean and the variance
(i.e., the fluctuations). Given that now not only the variance but the whole
distribution of the realized SNR is available, the previous proposed approach
based on the first two moments could also be extended to the optimization of a
given quantile by borrowing techniques from robust regression and robust
statistics. This, again, is a far from trivial problem which deserves a line
of research on its own.

On a final note, we recall how the asymptotic analysis can shed some light on
the convergence properties of the SNR, when the noise includes the
contribution from interfering sources. Using a simplified version of Theorem
\ref{theoremCLTs} for time-uncorrelated sources, it was theoretically shown in
\cite{ML05B,ML06} that, in scenarios where interferences are much more
powerful than the background noise, the minimum number of snapshots per
antenna to achieve an output SNR within $3$dB of the optimum one becomes: i)
$N>2K$ in the supervised case (compare with the classical $N>2M$ of the rule
proposed in \cite{RMB74}); ii) $N>\left(  2+\mathsf{SNR}_{\mathsf{opt}%
}\right)  K$ in the unsupervised case, where $K$ is the dimension of the
interference subspace. Hence, diagonal loading reduces the number of needed
samples by approximately a factor of $K/M$ (relative interference subspace dimension).

\section{Mathematical tools and preparatory results%
\label{sectionGauss}%
}

In this section, we introduce some mathematical tools and intermediate
technical results that will be useful for the proof of the central limit
theorems in Section \ref{sectionMain}. In the sequel, we will denote by
$\mathbf{Z}_{M}\in\mathbb{C}^{M\times M}$ and $\mathbf{\tilde{Z}}_{N}%
\in\mathbb{C}^{N\times N}$ sequences of arbitrary diagonal nonrandom matrices
with uniformly bounded spectral norm (in $M$ and $N$, respectively).
Similarly, $\mathbf{\Theta}_{M}\in\mathbb{C}^{M\times M}$ and $\mathbf{\tilde
{\Theta}}_{N}\in\mathbb{C}^{N\times N}$ will represent sequences of positive
definite nonrandom matrices having trace norm uniformly bounded from above by
finite scalars denoted, respectively, by $\left\Vert \mathbf{\Theta
}\right\Vert _{\operatorname*{tr},\sup}$ and $\left\Vert \mathbf{\tilde
{\Theta}}\right\Vert _{\operatorname*{tr},\sup}$, and trace operator uniformly
bounded away from zero, i.e., $\min\left\{  \theta_{\inf},\tilde{\theta}%
_{\inf}\right\}  >0$, where $\theta_{\inf}=\inf_{M\geq1}\operatorname*{tr}%
\left[  \mathbf{\Theta}_{M}\right]  $ and $\tilde{\theta}_{\inf}=\inf_{N\geq
1}\operatorname*{tr}\left[  \mathbf{\tilde{\Theta}}_{N}\right]  $. In
particular, notice that $\left\Vert \mathbf{\Theta}_{M}\right\Vert _{F}%
\leq\left\Vert \mathbf{\Theta}_{M}\right\Vert _{\operatorname*{tr}}$, and so
the Frobenius norm of $\mathbf{\Theta}_{M}$ is also uniformly bounded. For
instance, in the cases $\mathbf{\Theta}_{M}=\frac{1}{M}\mathbf{Z}_{M}^{\ast
}\mathbf{Z}_{M}$ and $\mathbf{\Theta}_{M}=\mathbf{u}_{M}\mathbf{u}_{M}^{\ast}%
$, we have $\left\Vert \frac{1}{M}\mathbf{Z}_{M}^{\ast}\mathbf{Z}%
_{M}\right\Vert _{F}=\frac{1}{M^{1/2}}\left(  \frac{1}{M}\operatorname*{tr}%
\left[  \left(  \mathbf{Z}_{M}^{\ast}\mathbf{Z}_{M}\right)  ^{2}\right]
\right)  ^{1/2}=\mathcal{O}\left(  N^{-1/2}\right)  $ and $\left\Vert
\mathbf{u}_{M}\mathbf{u}_{M}^{\ast}\right\Vert _{F}=\left\Vert \mathbf{u}%
_{M}\right\Vert ^{2}=\mathcal{O}\left(  1\right)  $, respectively. We remark
that the positive definiteness of the matrices $\mathbf{\Theta}_{M}$ and
$\mathbf{\tilde{\Theta}}_{N}$ only represents a purely technical assumption
that will facilitate the proofs, but which can be relaxed to extend the
results to the case of arbitrary not necessarily positive definite matrices.

Next, we introduce some results that will represent a set of essential tools
for the proof of Theorem \ref{theoremCLTs} and Theorem \ref{theoremCLTu}.

\subsection{Gaussian tools}

We first briefly comment on the bounded character of the empirical moments of
the spectral norm. Let $p$ be a fixed integer and let $\left\{  \mathbf{\tilde
{Z}}_{N}^{(l)}\right\}  $, $1\leq l\leq p$, denote a set of $p$ sequences of
$N\times N$ diagonal deterministic matrices with uniformly bounded spectral
norm in $N$. Then, for $p\geq1$, we have
\begin{equation}
\mathbb{E}\left[  \left\Vert \frac{\mathbf{X}_{N}\mathbf{\tilde{Z}}_{N}%
^{(1)}\mathbf{X}{}_{N}^{\ast}}{N}\frac{\mathbf{X}_{N}\mathbf{\tilde{Z}}%
_{N}^{(2)}\mathbf{X}{}_{N}^{\ast}}{N}\cdots\frac{\mathbf{X}_{N}\mathbf{\tilde
{Z}}_{N}^{(p)}\mathbf{X}{}_{N}^{\ast}}{N}\right\Vert \right]  <K_{p}%
.\label{eq_bounded_moments}%
\end{equation}
The proof of (\ref{eq_bounded_moments}) follows by first writing, using the
submultiplicative property of the spectral norm,%
\begin{multline*}
\mathbb{E}\left[  \left\Vert \frac{\mathbf{X}_{N}\mathbf{\tilde{Z}}_{N}%
^{(1)}\mathbf{X}{}_{N}^{\ast}}{N}\frac{\mathbf{X}_{N}\mathbf{\tilde{Z}}%
_{N}^{(2)}\mathbf{X}{}_{N}^{\ast}}{N}\cdots\frac{\mathbf{X}_{N}\mathbf{\tilde
{Z}}_{N}^{(p)}\mathbf{X}{}_{N}^{\ast}}{N}\right\Vert \right]  \leq\\
\leq\left(  \prod\limits_{r=1}^{p}\sup_{N\geq1}\left\Vert \mathbf{\tilde{Z}%
}_{N}^{(r)}\right\Vert \right)  \mathbb{E}\left[  \left\Vert \frac
{\mathbf{X}_{N}\mathbf{X}_{N}^{\ast}}{N}\right\Vert ^{p}\right]  \text{.}%
\end{multline*}
and then applying the following intermediate result.%

\begin{lemma}%
\label{momsn}%
Let $\mathbf{\tilde{X}}_{N}\in\mathbb{R}^{M\times N}$ be a matrix having
entries defined as i.i.d. Gaussian random variables with mean zero and
variance one. Then, the following inequality holds for every $q\geq1$, i.e.,%
\[
\sup_{N\geq1}\mathbb{E}\left[  \left\Vert \frac{\mathbf{\tilde{X}}_{N}}%
{\sqrt{N}}\right\Vert ^{q}\right]  <+\infty\text{.}%
\]%
\end{lemma}%
\begin{proof}%
The proof is based on some well-known results about the concentration of
Gaussian measures and its applications to random matrix theory (see, e.g.,
\cite{L07B}). In particular, we build upon the following large deviation
inequality for the largest singular value of a Gaussian matrix \cite[Theorem
II.13]{DS01B}, namely,%
\begin{equation}
\mathbb{P}\left(  \left\vert \left\Vert \mathbf{\tilde{X}}_{N}\right\Vert
-\left(  \sqrt{M}+\sqrt{N}\right)  \right\vert \geq t\right)  <2\exp\left(
-\frac{t^{2}}{2}\right)  \text{,}\label{conmes}%
\end{equation}
for any $t>0$. Furthermore, for every non-negative random variable $X$, we
have $\mathbb{E}\left[  X\right]  =\int_{0}^{\infty}\mathbb{P}\left(  X\geq
x\right)  \ dx$. Now, using the change of variables $x=t^{q}$, $dx=qt^{q-1}%
\ dt$, notice that $\mathbb{E}\left[  X^{q}\right]  =\int_{0}^{\infty
}\mathbb{P}\left(  X^{q}\geq x\right)  \ dx=\int_{0}^{\infty}\mathbb{P}\left(
X\geq t\right)  \ qt^{q-1}\ dt$. Finally, letting $X=\left\vert \left\Vert
\mathbf{\tilde{X}}_{N}\right\Vert -\left(  \sqrt{M}+\sqrt{N}\right)
\right\vert $, we get from (\ref{conmes})%
\[
\mathbb{E}\left[  \left\vert \left\Vert \mathbf{\tilde{X}}_{N}\right\Vert
-\left(  \sqrt{M}+\sqrt{N}\right)  \right\vert ^{q}\right]  \leq2q\int
_{0}^{\infty}e^{-\frac{1}{2}t^{2}}t^{q-1}\ dt=2^{q/2}q\Gamma\left(  \frac
{q}{2}\right)  \leq q^{q/2+1}\text{,}%
\]
where $\Gamma\left(  x\right)  $ is the Gamma function, and we conclude that
\[
\mathbb{E}\left[  \left\Vert \frac{\mathbf{\tilde{X}}_{N}}{\sqrt{N}%
}\right\Vert ^{q}\right]  =K_{q}+\mathcal{O}\left(  N^{-1/2}\right)  \text{.}%
\]%
\end{proof}%
Indeed, if we let $\mathbf{X}_{N}=\frac{1}{\sqrt{2}}\mathbf{\tilde{X}}%
_{N}^{\left(  re\right)  }+\operatorname*{i}\frac{1}{\sqrt{2}}\mathbf{\tilde
{X}}_{N}^{\left(  im\right)  }$, where the matrices $\mathbf{\tilde{X}}%
_{N}^{\left(  re\right)  }$ and $\mathbf{\tilde{X}}_{N}^{\left(  im\right)  }$
are independently defined as the matrix $\mathbf{\tilde{X}}_{N}$, then,
applying Jensen's inequality along with Lemma \ref{momsn}, we get%
\[
\mathbb{E}\left[  \left\Vert \frac{\mathbf{X}_{N}}{\sqrt{N}}\right\Vert
^{q}\right]  \leq2^{r/2-1}\left(  \mathbb{E}\left[  \left\Vert \frac
{\mathbf{\tilde{X}}_{N}^{\left(  re\right)  }}{\sqrt{N}}\right\Vert
^{q}\right]  +\mathbb{E}\left[  \left\Vert \frac{\mathbf{\tilde{X}}%
_{N}^{\left(  im\right)  }}{\sqrt{N}}\right\Vert ^{q}\right]  \right)
<K_{q}\text{,}%
\]
and (\ref{eq_bounded_moments}) follows finally by taking $q=2p$.

We now introduce two further tools; with some abuse of notation, let
$\Gamma=\Gamma\left(  \mathbf{X}_{N},\mathbf{X}{}_{N}^{\ast}\right)  $ be a
$\mathcal{C}^{1}$ complex function such that both itself and its derivatives
are polynomically bounded. Following the approach in \cite{HKLNP08}, in our
proof of the CLT we will make intensive use of the \textit{Nash-Poincar\'{e}
inequality}, i.e.,%
\begin{equation}
\operatorname*{var}\left(  \Gamma\left(  \mathbf{X}_{N},\mathbf{X}{}_{N}%
^{\ast}\right)  \right)  \leq\sum_{i=1}^{M}\sum_{j=1}^{N}\mathbb{E}\left[
\left\vert \frac{\partial\Gamma\left(  \mathbf{X}_{N},\mathbf{X}{}_{N}^{\ast
}\right)  }{\partial X_{ij}}\right\vert ^{2}+\left\vert \frac{\partial
\Gamma\left(  \mathbf{X}_{N},\mathbf{X}{}_{N}^{\ast}\right)  }{\partial
\overline{X_{ij}}}\right\vert ^{2}\right]  \text{,}\label{eq_Nash_Poincare}%
\end{equation}
where the upper bar denotes complex conjugation, as well as the
\textit{integration by parts formula for Gaussian functionals}, namely%
\begin{equation}
\mathbb{E}\left[  X_{ij}\Gamma\left(  \mathbf{X}_{N},\mathbf{X}{}_{N}^{\ast
}\right)  \right]  =\mathbb{E}\left[  \frac{\partial\Gamma\left(
\mathbf{X}_{N},\mathbf{X}{}_{N}^{\ast}\right)  }{\partial\overline{X_{ij}}%
}\right]  \text{.}\label{eq_integration_by_parts}%
\end{equation}

\subsection{Variance controls and estimates of expected values}

Let us define the random variables%
\begin{equation}
\Phi_{M}^{\left(  k\right)  }=\Phi_{M}^{\left(  k\right)  }\left(
\mathbf{X}_{N}\right)  =\operatorname*{tr}\left[  \mathbf{\Theta}%
_{M}\mathbf{Q}_{M}^{k}\right]  ,\ \Psi_{M}^{\left(  k\right)  }=\Psi
_{M}^{\left(  k\right)  }\left(  \mathbf{X}_{N}\right)  =\operatorname*{tr}%
\left[  \mathbf{\Theta}_{M}\mathbf{Q}_{M}^{k}\frac{\mathbf{X}_{N}%
\mathbf{\tilde{Z}}_{N}\mathbf{X}_{N}^{\ast}}{N}\right] \label{RVdef}%
\end{equation}
where $k$ is a finite positive integer. The proof of the following variance
estimates essentially rely on the Nash-Poincar\'{e} inequality in
(\ref{eq_Nash_Poincare}).%

\begin{lemma}%
\label{lemmaVC_Phi_Psi}%
With all above definitions, the following variance controls hold:%
\[
\operatorname*{var}\left(  \Phi_{M}^{\left(  k\right)  }\left(  \mathbf{X}%
_{N}\right)  \right)  =\mathcal{O}\left(  \frac{\left\Vert \mathbf{\Theta
}\right\Vert _{F}^{2}}{N}\right)  \text{,}%
\]
and%
\[
\operatorname*{var}\left(  \Psi_{M}^{\left(  k\right)  }\left(  \mathbf{X}%
_{N}\right)  \right)  =\mathcal{O}\left(  \frac{\left\Vert \mathbf{\Theta
}\right\Vert _{F}^{2}}{N}\right)  \text{.}%
\]%
\end{lemma}%
\begin{proof}%
See Appendix \ref{sectionAppVC}.%
\end{proof}%

Also of particular use in our derivations will be the following approximation
rules, whose proof has been postponed to Appendix \ref{sectionAppEV}.%

\begin{proposition}%
\label{theoremEV1}%
With all above definitions, the following expectations hold, namely%
\[
\mathbb{E}\left[  \Phi_{M}^{\left(  1\right)  }\left(  \mathbf{X}_{N}\right)
\right]  =\operatorname*{tr}\left[  \mathbf{\Theta E}\right]  +\mathcal{O}%
\left(  \frac{\left\Vert \mathbf{\Theta}\right\Vert _{F}}{N^{3/2}}\right)
\]
and%
\[
\mathbb{E}\left[  \Psi_{M}^{\left(  1\right)  }\left(  \mathbf{X}_{N}\right)
\right]  =\frac{1}{N}\operatorname*{tr}\left[  \mathbf{\tilde{Z}}\left(
\mathbf{I}_{N}+\delta_{M}\mathbf{T}\right)  ^{-1}\right]  \operatorname*{tr}%
\left[  \mathbf{\Theta E}\right]  +\mathcal{O}\left(  \frac{\left\Vert
\mathbf{\Theta}\right\Vert _{F}}{N^{3/2}}\right)  \text{.}%
\]%
\end{proposition}%
%

\begin{proposition}%
\label{theoremEV2}%
With all above definitions, the following expectations hold, namely%
\[
\mathbb{E}\left[  \Phi_{M}^{\left(  2\right)  }\left(  \mathbf{X}_{N}\right)
\right]  =\frac{1}{1-\gamma\tilde{\gamma}}\operatorname*{tr}\left[
\mathbf{\Theta E}^{2}\right]  +\mathcal{O}\left(  \frac{\left\Vert
\mathbf{\Theta}\right\Vert _{F}}{N^{3/2}}\right)
\]
and%
\begin{align*}
\mathbb{E}\left[  \Psi_{M}^{\left(  2\right)  }\left(  \mathbf{X}_{N}\right)
\right]   &  =\frac{1}{N}\operatorname*{tr}\left[  \mathbf{\tilde{Z}}\left(
\mathbf{I}_{N}+\delta_{M}\mathbf{T}\right)  ^{-1}\right]  \frac{1}%
{1-\gamma\tilde{\gamma}}\operatorname*{tr}\left[  \mathbf{\Theta E}^{2}\right]
\\
&  -\frac{1}{N}\operatorname*{tr}\left[  \mathbf{\tilde{E}\tilde{Z}}\left(
\mathbf{I}_{N}+\delta_{M}\mathbf{T}\right)  ^{-1}\right]  \frac{\gamma
}{1-\gamma\tilde{\gamma}}\operatorname*{tr}\left[  \mathbf{\Theta E}\right]
+\mathcal{O}\left(  \frac{\left\Vert \mathbf{\Theta}\right\Vert _{F}}{N^{3/2}%
}\right)  \text{.}%
\end{align*}%
\end{proposition}%
%

\begin{proposition}%
\label{theoremEV3}%
With all above definitions, the following expectations hold, namely%
\begin{multline*}
\mathbb{E}\left[  \Phi_{M}^{\left(  3\right)  }\left(  \mathbf{X}_{N}\right)
\right]  =\frac{1}{\left(  1-\gamma\tilde{\gamma}\right)  ^{3}}\left(
\tilde{\gamma}\frac{1}{N}\operatorname*{tr}\left[  \mathbf{E}^{3}\right]
-\gamma^{2}\frac{1}{N}\operatorname*{tr}\left[  \mathbf{\tilde{E}}^{3}\right]
\right)  \operatorname*{tr}\left[  \mathbf{\Theta E}^{2}\right]  +\\
+\frac{1}{\left(  1-\gamma\tilde{\gamma}\right)  ^{2}}\operatorname*{tr}%
\left[  \mathbf{\Theta E}^{3}\right]  +\mathcal{O}\left(  \frac{\left\Vert
\mathbf{\Theta}\right\Vert _{F}}{N^{3/2}}\right)
\end{multline*}
and%
\begin{multline*}
\mathbb{E}\left[  \Psi_{M}^{\left(  3\right)  }\left(  \mathbf{X}_{N}\right)
\right]  =\\
=\frac{\operatorname*{tr}\left[  \mathbf{\tilde{Z}}\left(  \mathbf{I}%
_{N}+\delta_{M}\mathbf{T}\right)  ^{-1}\right]  }{N\left(  1-\gamma
\tilde{\gamma}\right)  ^{2}}\left\{  \frac{\tilde{\gamma}\frac{1}%
{N}\operatorname*{tr}\left[  \mathbf{E}^{3}\right]  -\gamma^{2}\frac{1}%
{N}\operatorname*{tr}\left[  \mathbf{\tilde{E}}^{3}\right]  }{1-\gamma
\tilde{\gamma}}\operatorname*{tr}\left[  \mathbf{\Theta E}^{2}\right]
+\operatorname*{tr}\left[  \mathbf{\Theta E}^{3}\right]  \right\} \\
-\frac{\operatorname*{tr}\left[  \mathbf{\tilde{Z}\tilde{E}}\left(
\mathbf{I}_{N}+\delta_{M}\mathbf{T}\right)  ^{-1}\right]  }{N\left(
1-\gamma\tilde{\gamma}\right)  ^{2}}\left\{  \frac{\frac{1}{N}%
\operatorname*{tr}\left[  \mathbf{E}^{3}\right]  -\gamma^{3}\frac{1}%
{N}\operatorname*{tr}\left[  \mathbf{\tilde{E}}^{3}\right]  }{1-\gamma
\tilde{\gamma}}\operatorname*{tr}\left[  \mathbf{\Theta E}\right]
+\gamma\operatorname*{tr}\left[  \mathbf{\Theta E}^{2}\right]  \right\} \\
+\left(  \frac{\gamma}{1-\gamma\tilde{\gamma}}\right)  ^{2}\frac{1}%
{N}\operatorname*{tr}\left[  \mathbf{\tilde{Z}\tilde{E}}^{2}\left(
\mathbf{I}_{N}+\delta_{M}\mathbf{T}\right)  ^{-1}\right]  \operatorname*{tr}%
\left[  \mathbf{\Theta E}\right]  +\mathcal{O}\left(  \frac{\left\Vert
\mathbf{\Theta}\right\Vert _{F}}{N^{3/2}}\right)  \text{.}%
\end{multline*}%
\end{proposition}%
%

\begin{proposition}%
\label{theoremEV4}%
With all above definitions, the following expectation holds, namely%
\begin{multline*}
\mathbb{E}\left[  \Phi_{M}^{\left(  4\right)  }\left(  \mathbf{X}_{N}\right)
\right]  =\frac{1}{\left(  1-\gamma\tilde{\gamma}\right)  ^{3}}%
\operatorname*{tr}\left[  \mathbf{\Theta E}^{4}\right]  +\\
+\frac{2\operatorname*{tr}\left[  \mathbf{\Theta E}^{3}\right]  }{\left(
1-\gamma\tilde{\gamma}\right)  ^{4}}\left\{  \tilde{\gamma}\frac{1}%
{N}\operatorname*{tr}\left[  \mathbf{E}^{3}\right]  -\gamma^{2}\frac{1}%
{N}\operatorname*{tr}\left[  \mathbf{\tilde{E}}^{3}\right]  \right\} \\
+\frac{\operatorname*{tr}\left[  \mathbf{\Theta E}^{2}\right]  }{\left(
1-\gamma\tilde{\gamma}\right)  ^{4}}\left\{  \gamma^{3}\frac{1}{N}%
\operatorname*{tr}\left[  \mathbf{\tilde{E}}^{4}\right]  +\tilde{\gamma}%
\frac{1}{N}\operatorname*{tr}\left[  \mathbf{E}^{4}\right]  \right\}
+\frac{2\operatorname*{tr}\left[  \mathbf{\Theta E}^{2}\right]  }{\left(
1-\gamma\tilde{\gamma}\right)  ^{5}}\times\\
\times\left\{  \gamma^{4}\left(  \frac{1}{N}\operatorname*{tr}\left[
\mathbf{\tilde{E}}^{3}\right]  \right)  ^{2}+\tilde{\gamma}^{2}\left(
\frac{1}{N}\operatorname*{tr}\left[  \mathbf{E}^{3}\right]  \right)
^{2}-\gamma\left(  1+\gamma\tilde{\gamma}\right)  \frac{1}{N}%
\operatorname*{tr}\left[  \mathbf{\tilde{E}}^{3}\right]  \frac{1}%
{N}\operatorname*{tr}\left[  \mathbf{E}^{3}\right]  \right\}  +\\
+\mathcal{O}\left(  \frac{\left\Vert \mathbf{\Theta}\right\Vert _{F}}{N^{3/2}%
}\right)  \text{.}%
\end{multline*}%
\end{proposition}%

\section{Elements of the proof of the asymptotic Gaussianity of the SNR%
\label{sectionCLT}%
}

Let us consider the real-valued random variable $\xi_{M}=A_{M}\sqrt{N}\left(
a_{M}-\overline{a}_{M}\right)  +B_{M}\sqrt{N}\left(  b_{M}-\overline{b}%
_{M}\right)  $, where $a_{M},\overline{a}_{M},b_{M},\overline{b}_{M}$ are
defined in (\ref{eq_defab})-(\ref{eq_def_abbar}) and where $A_{M}$ and $B_{M}
$ are two real-valued nonrandom coefficients bounded above for all $M$ by
constants $A_{\sup}$ and $B_{\sup}$, respectively. In particular, notice that
if $\mathbf{\Theta}_{M}=\mathbf{u}_{M}\mathbf{u}_{M}^{\mathbf{\ast}}$ then we
have $\Phi_{M}^{\left(  1\right)  }=a_{M}$ and $\Phi_{M}^{\left(  2\right)
}=b_{M}$, and also $\bar{\Phi}_{M}^{\left(  1\right)  }=\bar{a}_{M}$ and
$\bar{\Phi}_{M}^{\left(  2\right)  }=\bar{b}_{M}$. We begin this section by
stating a theorem that establishes a CLT for the fluctuations of $\xi_{M} $,
and which will be instrumental in proving Theorem \ref{theoremCLTs} and
Theorem \ref{theoremCLTu}.%

\begin{theorem}%
\label{theoremBasicCLT}%
Assume that $\left[  A_{M},B_{M}\right]  $ is a deterministic real-valued
vector whose norm is uniformly bounded above and below. Then, under
$(\mathbf{As1}-\mathbf{As3})$, the following CLT holds:%
\begin{equation}
\sqrt{N}\sigma_{\xi,M}^{-1}\left(  A_{M},B_{M}\right)  \left(  A_{M}\left(
a_{M}-\overline{a}_{M}\right)  +B_{M}\left(  b_{M}-\overline{b}_{M}\right)
\right)  \overset{\mathcal{L}}{\rightarrow}\mathcal{N}\left(  0,1\right)
\text{,}\label{xiCLT}%
\end{equation}
where $\sigma_{\xi,M}^{2}\left(  A_{M},B_{M}\right)  =\left[
\begin{array}
[c]{cc}%
A_{M} & B_{M}%
\end{array}
\right]  \mathbf{\Sigma}_{M}\left[
\begin{array}
[c]{cc}%
A_{M} & B_{M}%
\end{array}
\right]  ^{T}$, with $\mathbf{\Sigma}_{M}$ being a real-valued symmetric
positive definite matrix having entries $\left[  \mathbf{\Sigma}_{M}\right]
_{1,1}=\sigma_{M,a^{2}}$, $\left[  \mathbf{\Sigma}_{M}\right]  _{2,2}%
=\sigma_{M,b^{2}}$, and $\left[  \mathbf{\Sigma}_{M}\right]  _{1,2}=\left[
\mathbf{\Sigma}_{M}\right]  _{2,1}=\sigma_{M,ab}=\sigma_{M,ba}$, given by%
\begin{align}
\sigma_{M,a^{2}} &  =\frac{\tilde{\gamma}}{1-\gamma\tilde{\gamma}}\left(
\mathbf{u}_{M}^{\mathbf{\ast}}\mathbf{E}_{M}^{2}\mathbf{u}_{M}\right)
^{2}\text{,}\label{sigmaA2}\\
\sigma_{M,ab} &  =\sigma_{M,ba}=\frac{2\tilde{\gamma}}{\left(  1-\gamma
\tilde{\gamma}\right)  ^{2}}\mathbf{u}_{M}^{\mathbf{\ast}}\mathbf{E}_{M}%
^{2}\mathbf{u}_{M}\mathbf{u}_{M}^{\mathbf{\ast}}\mathbf{E}_{M}^{3}%
\mathbf{u}_{M}\nonumber\\
&  +\frac{\left(  \mathbf{u}_{M}^{\mathbf{\ast}}\mathbf{E}_{M}^{2}%
\mathbf{u}_{M}\right)  ^{2}}{\left(  1-\gamma\tilde{\gamma}\right)  ^{3}%
}\left\{  \tilde{\gamma}^{2}\frac{1}{N}\operatorname*{tr}\left[
\mathbf{E}_{M}^{3}\right]  -\gamma\frac{1}{N}\operatorname*{tr}\left[
\mathbf{\tilde{E}}_{N}^{3}\right]  \right\}  \text{,}\label{sigmaAB}%
\end{align}
and%
\begin{multline}
\sigma_{M,b^{2}}=\frac{2\tilde{\gamma}}{\left(  1-\gamma\tilde{\gamma}\right)
^{3}}\mathbf{u}_{M}^{\mathbf{\ast}}\mathbf{E}_{M}^{4}\mathbf{u}_{M}%
\mathbf{u}_{M}^{\mathbf{\ast}}\mathbf{\mathbf{E}}_{M}^{2}\mathbf{u}_{M}%
+\frac{2\tilde{\gamma}}{\left(  1-\gamma\tilde{\gamma}\right)  ^{3}}\left(
\mathbf{u}_{M}^{\mathbf{\ast}}\mathbf{E}_{M}^{3}\mathbf{u}_{M}\right)  ^{2}\\
+\frac{4\mathbf{u}_{M}^{\mathbf{\ast}}\mathbf{E}_{M}^{3}\mathbf{u}%
_{M}\mathbf{u}_{M}^{\mathbf{\ast}}\mathbf{E}_{M}^{2}\mathbf{u}_{M}}{\left(
1-\gamma\tilde{\gamma}\right)  ^{4}}\left\{  \tilde{\gamma}^{2}\frac{1}%
{N}\operatorname*{tr}\left[  \mathbf{E}_{M}^{3}\right]  -\gamma\frac{1}%
{N}\operatorname*{tr}\left[  \mathbf{\tilde{E}}_{N}^{3}\right]  \right\} \\
+\frac{\left(  \mathbf{u}_{M}^{\mathbf{\ast}}\mathbf{E}_{M}^{2}\mathbf{u}%
_{M}\right)  ^{2}}{\left(  1-\gamma\tilde{\gamma}\right)  ^{4}}\left\{
\tilde{\gamma}^{2}\frac{1}{N}\operatorname*{tr}\left[  \mathbf{E}_{M}%
^{4}\right]  +\gamma^{2}\frac{1}{N}\operatorname*{tr}\left[  \mathbf{\tilde
{E}}_{N}^{4}\right]  \right\}  +\frac{2\left(  \mathbf{u}_{M}^{\mathbf{\ast}%
}\mathbf{E}_{M}^{2}\mathbf{u}_{M}\right)  ^{2}}{\left(  1-\gamma\tilde{\gamma
}\right)  ^{5}}\times\\
\times\left\{  \tilde{\gamma}^{3}\left(  \frac{1}{N}\operatorname*{tr}\left[
\mathbf{E}_{M}^{3}\right]  \right)  ^{2}-2\gamma\tilde{\gamma}\frac{1}%
{N}\operatorname*{tr}\left[  \mathbf{E}_{M}^{3}\right]  \frac{1}%
{N}\operatorname*{tr}\left[  \mathbf{\tilde{E}}_{N}^{3}\right]  +\gamma
^{3}\left(  \frac{1}{N}\operatorname*{tr}\left[  \mathbf{\tilde{E}}_{N}%
^{3}\right]  \right)  ^{2}\right\} \label{sigmaB2}%
\end{multline}%
\end{theorem}%
\begin{proof}%
Define $\Psi_{M}\left(  \omega\right)  =\exp\left(  \operatorname*{i}\omega
\xi_{M}\right)  $, and let $\mathbb{E}\left[  \Psi_{M}\left(  \omega\right)
\right]  $ be the characteristic function of $\xi_{M}$. The proof of Theorem
\ref{theoremBasicCLT} is based on Levy's continuity theorem, which allows us
to prove convergence in distribution by showing point-wise convergence of
characteristic functions \cite{B95B}. More specifically, similarly as in
\cite{HKLNP08}, we study weak convergence to a Gaussian law by showing%
\[
\mathbb{E}\left[  \Psi_{M}\left(  \omega\right)  \right]  -\exp\left(
-\frac{\omega^{2}}{2}\sigma_{\xi,M}^{2}\left(  A_{M},B_{M}\right)  \right)
\underset{M,N\rightarrow\infty}{\rightarrow}0\text{.}%
\]
In particular, we show that%
\begin{equation}
\frac{\partial}{\partial\omega}\mathbb{E}\left[  \Psi_{M}\left(
\omega\right)  \right]  =-\omega\sigma_{\xi,M}^{2}\left(  A_{M},B_{M}\right)
\mathbb{E}\left[  \Psi_{M}\left(  \omega\right)  \right]  +R_{N}\left(
\omega\right)  \text{,}\label{WWP}%
\end{equation}
where $R_{N}\left(  \omega\right)  $ is an error term vanishing asymptotically
as $N\rightarrow\infty$ uniformly in $\omega$ on compact subsets. In order to
prove (\ref{WWP}), we proceed by differentiating the characteristic function
as%
\begin{multline*}
\frac{\partial}{\partial\omega}\mathbb{E}\left[  \Psi_{M}\left(
\omega\right)  \right]  =\operatorname*{i}\mathbb{E}\left[  \xi_{M}\Psi
_{M}\left(  \omega\right)  \right] \\
=\operatorname*{i}A_{M}\sqrt{N}\mathbb{E}\left[  \left(  a_{M}-\overline
{a}_{M}\right)  \Psi_{M}\left(  \omega\right)  \right]  +\operatorname*{i}%
B_{M}\sqrt{N}\mathbb{E}\left[  \left(  b_{M}-\overline{b}_{M}\right)  \Psi
_{M}\left(  \omega\right)  \right]  \text{.}%
\end{multline*}

The following proposition provides the computation of the expectation
$\mathbb{E}\left[  \xi_{M}\Psi_{M}\left(  \omega\right)  \right]  $; see
Appendix \ref{AppProposition_aM_bM} for a proof.%
\end{proof}%
%

\begin{proposition}%
\label{proposition_aM_bM}%
With the above definitions, the following expectations hold, namely%
\begin{equation}
\sqrt{N}\mathbb{E}\left[  \left(  a_{M}-\overline{a}_{M}\right)  \Psi
_{M}\left(  \omega\right)  \right]  =\operatorname*{i}\omega\left(
A\sigma_{a^{2}}+B\sigma_{ab}\right)  \mathbb{E}\left[  \Psi\left(
\omega\right)  \right]  +\mathcal{O}\left(  N^{-1/2}\right)  \text{,}%
\label{prop_aM}%
\end{equation}
and%
\begin{equation}
\sqrt{N}\mathbb{E}\left[  \left(  b_{M}-\overline{b}_{M}\right)  \Psi
_{M}\left(  \omega\right)  \right]  =\operatorname*{i}\omega\left(
A\sigma_{ba}+B\sigma_{b^{2}}\right)  \mathbb{E}\left[  \Psi\left(
\omega\right)  \right]  +\mathcal{O}\left(  N^{-1/2}\right)  \text{.}%
\label{prop_bM}%
\end{equation}
Moreover, the term $\mathcal{O}\left(  N^{-1/2}\right)  $ depends neither on
the coefficients $A_{M}$ and $B_{M}$ nor on $\omega$, assuming that this last
parameter takes values on a bounded interval.%
\end{proposition}%

Therefore, we have (recall that $\sigma_{M,ab}=\sigma_{M,ba}$)%
\begin{equation}
\mathbb{E}\left[  \xi_{M}\Psi_{M}\left(  \omega\right)  \right]
=\operatorname*{i}\omega\left(  A_{M}^{2}\sigma_{M,a^{2}}+2A_{M}B_{M}%
\sigma_{M,ab}+B_{M}^{2}\sigma_{M,b^{2}}\right)  \mathbb{E}\left[  \Psi\left(
\omega\right)  \right]  +\mathcal{O}\left(  N^{-1/2}\right)  \text{.}%
\label{expectation_comp}%
\end{equation}
Furthermore, a sufficient and necessary condition for the matrix
$\mathbf{\Sigma}_{M}$ to be positive definite is stated in the following
proposition (see Appendix \ref{AppProposition_sigmaAB} for a proof).%

\begin{proposition}%
\label{proposition_sigmaAB}%
Under the assumptions of Theorem \ref{theoremBasicCLT}, we have%
\[
0<\inf_{M\geq1}\sigma_{\xi,M}^{2}\left(  A_{M},B_{M}\right)  \leq\sup_{M\geq
1}\sigma_{\xi,M}^{2}\left(  A_{M},B_{M}\right)  <+\infty\text{.}%
\]

In order to complete the proof of Theorem \ref{theoremBasicCLT}, we need to
show that the sequence
\[
\left\{  \sigma_{\xi,M}^{-1}\left(  A_{M},B_{M}\right)  \left(  A_{M}\left(
a_{M}-\overline{a}_{M}\right)  +B_{M}\left(  b_{M}-\overline{b}_{M}\right)
\right)  \right\}  _{M}%
\]
is tight, and that every converging subsequence does it in distribution to a
standard Gaussian random variable. The proof of the previous two arguments
relies on Proposition \ref{proposition_sigmaAB} and follows along exactly the
same lines of that of Proposition 6 in \cite{HKLNP08}, and so we exclude it
from our exposition.%
\end{proposition}%
%

\begin{remark}%
\label{furtherApps}%
Theorem \ref{theoremBasicCLT} can be used to characterize the fluctuations of
the performance of optimal LMMSE or Wiener filters.\ Here, we particularly
mean the classical statistical problem of estimating the signal $s\left(
n\right)  $ in the linear signal model (\ref{DGP}) with $\beta=1$, by
minimizing the Bayesian mean-square error (MSE) risk. Specifically, recalling
that the MSE of a filter $\mathbf{w}$ is given by $\mathsf{MSE}\left(
\mathbf{w}\right)  =1-2\operatorname{Re}\left\{  \mathbf{w}^{H}\mathbf{s}%
\right\}  +\mathbf{w}^{H}\mathbf{Rw}$ (see, e.g., \cite{S91B,K93B}), we notice
that the asymptotic distribution of the MSE achieved by a sample
implementation of the optimal filter $\mathbf{w}_{\mathsf{LMMSE}}%
=\mathbf{R}_{M}^{-1}\mathbf{s}$ based on the covariance matrix estimator
(\ref{SCM}), and denoted by $\mathbf{\hat{w}}_{\mathsf{MSE}}$, can be readily
obtained by simply applying Theorem \ref{theoremBasicCLT} with $A_{M}%
=A_{mse,M}=2$ and $B_{M}=B_{mse,M}=1$ along with $\mathbf{\Theta}%
_{M}=\mathbf{u}_{M}\mathbf{u}_{M}^{\mathbf{\ast}}$, so that we get the random
variable:%
\[
\mathsf{MSE}\left(  \mathbf{\hat{w}}_{\mathsf{MSE}}\right)  -\overline
{\mathsf{MSE}\left(  \mathbf{\hat{w}}_{\mathsf{MSE}}\right)  }=A_{mse,M}%
\left(  a_{M}-\bar{a}_{M}\right)  +B_{mse,M}\left(  b_{M}-\bar{b}_{M}\right)
\text{.}%
\]
Related work on the study of the asymptotic Gaussianity of LMMSE receivers can
be found in \cite{KKHN09b}, where different techniques than used here based on
the martingale central limit theorem are considered without the assumption of
Gaussian observations. We notice that the problem above relies on a covariance
matrix which is unknown and therefore estimated, while in \cite{KKHN09b} the
authors rely on a given model of the covariance matrix itself, whose structure
is assumed to be known.%
\end{remark}%

We now complete the proof of Theorem \ref{theoremCLTs} and Theorem
\ref{theoremCLTu} by showing that, similarly as in Remark \ref{furtherApps},
the asymptotic distribution of the SNR performance measure under both
supervised and unsupervised training is given by Theorem \ref{theoremCLTs},
for sensible choices of the coefficients $A_{M}$ and $B_{M}$.

\subsection{Completing the proof of Theorem \ref{theoremCLTs} and Theorem
\ref{theoremCLTu}}

Let us define the following nonrandom coefficients:%
\begin{equation}
A_{s,M}=\frac{2\bar{a}_{M}}{\bar{b}_{M}}\text{,\quad}B_{s,M}=-\left(
\frac{\bar{a}_{M}}{\bar{b}_{M}}\right)  ^{2}\text{,}\label{AB_sM}%
\end{equation}
and%
\begin{equation}
A_{u,M}=\frac{2\bar{a}_{M}\bar{b}_{M}}{\left(  \bar{b}_{M}-\bar{a}_{M}%
^{2}\right)  ^{2}}\text{,\quad}B_{u,M}=-\left(  \frac{\bar{a}_{M}}{\bar{b}%
_{M}-\bar{a}_{M}^{2}}\right)  ^{2}\text{,}\label{AB_uM}%
\end{equation}
which are bounded above and away from zero uniformly in $M$ (cf. inequalities
(\ref{sup_ab}) - (\ref{bounds_ratio_bBar-aBar2}) in Appendix \ref{Appendix1}).
In particular, notice that%
\begin{equation}
\frac{A_{s,M}}{B_{s,M}}=\frac{A_{u,M}}{B_{u,M}}=-\frac{2\bar{b}_{M}}{\bar
{a}_{M}}\text{.}%
\end{equation}
Now, observe that we can write%

\begin{multline}
\sqrt{N}\left(  \mathsf{SNRs}\left(  \mathbf{\hat{w}}_{\mathsf{MVDR}}\right)
-\overline{\mathsf{SNRs}\left(  \mathbf{\hat{w}}_{\mathsf{MVDR}}\right)
}\right)  =\label{proofCLTst}\\
=\sqrt{N}\left(  A_{s,M}\left(  a_{M}-\bar{a}_{M}\right)  +B_{s,M}\left(
b_{M}-\bar{b}_{M}\right)  \right)  +\varepsilon_{s,M}\text{,}%
\end{multline}
and%
\begin{multline}
\sqrt{N}\left(  \mathsf{SNRu}\left(  \mathbf{\hat{w}}_{\mathsf{MVDR}}\right)
-\overline{\mathsf{SNRu}\left(  \mathbf{\hat{w}}_{\mathsf{MVDR}}\right)
}\right)  =\label{proofCLTut}\\
=\sqrt{N}\left(  A_{u,M}\left(  a_{M}-\bar{a}_{M}\right)  +B_{u,M}\left(
b_{M}-\bar{b}_{M}\right)  \right)  +\varepsilon_{u,M}\text{,}%
\end{multline}
where%
\[
\varepsilon_{s,M}=\sqrt{N}\left(  \frac{a_{M}}{b_{M}}-\frac{\bar{a}_{M}}%
{\bar{b}_{M}}\right)  ^{2}b_{M}\text{,}%
\]%
\begin{multline*}
\varepsilon_{u,M}=\sqrt{N}\left(  a_{M}-\bar{a}_{M}\right)  ^{2}\left(
\frac{\bar{a}_{M}}{\bar{b}_{M}-\bar{a}_{M}^{2}}\right)  ^{2}\\
+\sqrt{N}\left(  \frac{a_{M}}{b_{M}-a_{M}^{2}}-\frac{\bar{a}_{M}}{\bar{b}%
_{M}-\bar{a}_{M}^{2}}\right)  ^{2}\left(  b_{M}-a_{M}^{2}\right)
\equiv\varepsilon_{u,M}^{\left(  1\right)  }+\varepsilon_{u,M}^{\left(
2\right)  }\text{.}%
\end{multline*}
Next, we show that $\varepsilon_{s,M}=o_{p}\left(  1\right)  $ and
$\varepsilon_{u,M}=\varepsilon_{u,M}^{\left(  1\right)  }+\varepsilon
_{u,M}^{\left(  2\right)  }=o_{p}\left(  1\right)  $. Indeed, notice that we
can write%
\begin{multline*}
\mathbb{P}\left(  \left\vert \varepsilon_{s,M}\right\vert >\epsilon\right)
\leq\frac{\sqrt{N}}{\epsilon}\mathbb{E}\left[  \left\vert \frac{a_{M}}{b_{M}%
}-\frac{\bar{a}_{M}}{\bar{b}_{M}}\right\vert ^{2}\left\vert b_{M}\right\vert
\right] \\
\leq2\frac{\sqrt{N}}{\epsilon}\left(  \mathbb{E}\left[  \frac{\left\vert
a_{M}-\bar{a}_{M}\right\vert ^{2}}{\left\vert b_{M}\right\vert }\right]
+\frac{\left\vert \bar{a}_{M}\right\vert ^{2}}{\left\vert \bar{b}%
_{M}\right\vert ^{2}}\mathbb{E}\left[  \frac{\left\vert b_{M}-\bar{b}%
_{M}\right\vert ^{2}}{\left\vert b_{M}\right\vert }\right]  \right)  \text{,}%
\end{multline*}
where the last expression follows from Jensen's inequality. Let us further
define the random variables $\mathcal{X}_{M}=1/\left(  b_{M}-a_{M}^{2}\right)
$, and $\mathcal{X}_{1,M}=\mathcal{X}_{M}$,\ $\mathcal{X}_{2,M}=a_{M}%
^{2}\mathcal{X}_{M}$,\ $\mathcal{X}_{3,M}=a_{M}\mathcal{X}_{M}$,\ $\mathcal{X}%
_{4,M}=\mathcal{X}_{M}$, along with the nonrandom coefficients $C_{M}=\bar
{a}_{M}^{2}/\left(  \bar{b}_{M}-\bar{a}_{M}^{2}\right)  ^{2}$, and $C_{1,M}%
=1$, $C_{3,M}=2\bar{a}_{M}^{3}C_{M}$,\ $C_{k,M}=\bar{a}_{M}^{k}C_{M}$,
$k=2,4$. Then, we similarly have%
\[
\mathbb{P}\left(  \varepsilon_{u,M}^{\left(  1\right)  }>\epsilon\right)
\leq\frac{\sqrt{N}}{\epsilon}\left\vert C_{M}\right\vert \mathbb{E}\left[
\left\vert a_{M}-\bar{a}_{M}\right\vert ^{2}\right]
\]
and (notice that $\left(  a_{M}^{2}-\bar{a}_{M}^{2}\right)  ^{2}=\left(
a_{M}^{2}+2a_{M}\bar{a}_{M}+\bar{a}_{M}^{2}\right)  \left(  a_{M}-\bar{a}%
_{M}\right)  ^{2}$)%
\begin{multline*}
\mathbb{P}\left(  \varepsilon_{u,M}^{\left(  2\right)  }>\epsilon\right)
\leq\frac{\sqrt{N}}{\epsilon}\mathbb{E}\left[  \left\vert \frac{a_{M}}%
{b_{M}-a_{M}^{2}}-\frac{\bar{a}_{M}}{\bar{b}_{M}-\bar{a}_{M}^{2}}\right\vert
^{2}\left\vert b_{M}-a_{M}^{2}\right\vert \right] \\
\leq2\frac{\sqrt{N}}{\epsilon}\left\vert C_{M}\right\vert \mathbb{E}\left[
\left\vert \mathcal{X}_{M}\right\vert \left\vert b_{M}-\bar{b}_{M}\right\vert
^{2}\right]  +4\frac{\sqrt{N}}{\epsilon}\left(  \sum_{k=1}^{4}\left\vert
C_{k,M}\right\vert \mathbb{E}\left[  \left\vert \mathcal{X}_{k,M}\right\vert
\left\vert a_{M}-\bar{a}_{M}\right\vert ^{2}\right]  \right)  \text{,}%
\end{multline*}
Finally, using the bounds (\ref{sup_ab}) - (\ref{bounds_ratio_bBar-aBar2}) in
Appendix \ref{Appendix1} to show that%
\[
\max_{1\leq k\leq4}\sup_{M\geq1}\left\{  C_{k,M}\right\}  <+\infty,
\]
and
\[
\max_{1\leq k\leq4}\sup_{M\geq1}\left\{  \mathcal{X}_{k,M}\right\}  <+\infty
\]
with probability one, together with Jensen's inequality and Propositions
\ref{theoremEV1} and \ref{theoremEV2}, we conclude that both $\mathbb{P}%
\left(  \left\vert \varepsilon_{s,M}\right\vert >\epsilon\right)
\rightarrow0$ and $\mathbb{P}\left(  \left\vert \varepsilon_{u,M}\right\vert
>\epsilon\right)  \rightarrow0$ as $N\rightarrow\infty$. Hence, from
(\ref{proofCLTst}) and (\ref{proofCLTut}) along with the fact that
$\varepsilon_{s,M}=o_{p}\left(  1\right)  $ and $\varepsilon_{u,M}%
=o_{p}\left(  1\right)  $, we conclude that the central limit theorems in
Theorem \ref{theoremCLTs} and Theorem \ref{theoremCLTu} follow by Slutsky's
theorem and Theorem \ref{theoremCLTs} with $\sigma_{s,M}^{2}$ and
$\sigma_{u,M}^{2}$\ being given by the quadratic form $\sigma_{\xi,M}%
^{2}\left(  A_{M},B_{M}\right)  $, where the coefficients $A_{M}$ and $B_{M}$
are given by (\ref{AB_sM}) and (\ref{AB_uM}), respectively.

\section{Numerical validation%
\label{sectionSimulations}%
}

In this section, we compare the empirical distribution of the output SNR
obtained by simulations with the corresponding analytical expressions derived
in this paper. We considered a uniform linear array with elements located half
a wavelength apart. The exploration angle was $0$ deg. (desired signal), and
the array received interfering signals from the angles $-20$, $50$ and $55$
degrees. All signals were received at each antenna with power $10$dB above the
background noise. In this toy example, the time correlation matrix was fixed
to be a symmetric Toeplitz with its $n$th upper diagonal fixed to $e^{-n}$,
$n=0,\ldots,N-1$, and the diagonal loading parameter was fixed to $\alpha
=0.1$. In Figure \ref{figure1} and Figure \ref{figure2}, we represent the
measured histogram (bars) and asymptotic law (solid curves) of the output SNR
for different values of the parameters $M,N$, for both supervised and
unsupervised training, respectively. A total number of 10,000 realizations has
been considered to obtain the empirical probability density function. In each
figure, the upper plot corresponds to the case where the number of samples is
lower than the number of antennas, whereas in the lower plot we depict the
opposite situation. Observe that in both cases the asymptotic expressions give
a very accurate description of the fluctuations of the output SNR, even for
relatively low values of $M,N$. We also notice that the mismatch observed for
very low dimensions is readily corrected by slightly increasing $M$ and $N$.

\section{Conclusions%
\label{sectionConclusions}%
}

We have shown that the SNR of the diagonally loaded MVDR filters is
asymptotically Gaussian and have provided a closed-form expression for its
variance. A CLT has been established for the fluctuations of the SNR
performance of both supervised and unsupervised training methods. We resorted
to the Nash-Poincar\'{e} inequality and the integration by parts formula for
Gaussian functionals to derive variance and bias estimates for the
constituents of the SNR measure. In fact, the same elements describe also the
fluctuations of the mean-square error performance of this filter, which can be
written in terms of realized variance and bias, as well as of other optimal
linear filters, such as the Bayesian linear minimum mean-square error filter.
The results hold for Gaussian observations, but extensions based on a more
general integration by parts formula can be investigated for non-Gaussian observations.%

\appendix

\section{Further definitions and useful bounds%
\label{Appendix1}%
}

Throughout the appendices, we will use the following definitions, namely%
\[
\mathbf{\tilde{F}}=\mathbf{\mathbf{\mathbf{T}}}\left(  \mathbf{I}_{N}+\frac
{1}{N}\mathbb{E}\operatorname*{tr}\left[  \mathbf{Q}\right]  \mathbf{T}%
\right)  ^{-1}\text{,}%
\]
and also%
\[
\mathbf{F}=\left(  \frac{1}{N}\operatorname*{tr}\left[  \mathbf{\tilde{F}%
}\right]  \mathbf{I}_{M}+\alpha\mathbf{R}^{-1}\right)  ^{-1}\text{.}%
\]
Let $\mathbf{A}$ and $\mathbf{B}$ denote two arbitrary square complex
matrices. The following will be denoted in the sequel as \textit{resolvent
identity}, namely, $\mathbf{A}^{-1}-\mathbf{B}^{-1}=\mathbf{A}^{-1}\left(
\mathbf{B}-\mathbf{A}\right)  \mathbf{B}^{-1}$, where we have tacitly assumed
the invertibility of $\mathbf{A}$ and $\mathbf{B}$. In particular, using the
previous resolvent identity, we notice that%
\begin{equation}
\mathbf{Q}=\alpha^{-1}\mathbf{R}-\alpha^{-1}\mathbf{Q}\frac{1}{N}%
\mathbf{XTX}^{\ast}\mathbf{R}\text{.}\label{RI}%
\end{equation}
Furthermore, we define%
\[
\chi_{M}^{\left(  p\right)  }=\frac{1}{N}\operatorname*{tr}\left[
\mathbf{Q}^{p}\right]  -\frac{1}{N}\mathbb{E}\operatorname*{tr}\left[
\mathbf{Q}^{p}\right]
\]

Now, we introduce some inequalities that will be extensively used in our
derivations. First, let $X$ and $Y$ be two scalar and complex-valued random
variables having second-order moment. Then, we have%
\begin{equation}
\operatorname*{var}\left(  X+Y\right)  \leq\operatorname*{var}\left(
X\right)  +\operatorname*{var}\left(  Y\right)  +2\sqrt{\operatorname*{var}%
\left(  X\right)  \operatorname*{var}\left(  Y\right)  }\text{,}\label{varINE}%
\end{equation}
and also, from the Cauchy-Schwarz inequality,%
\begin{align}
\left\vert \mathbb{E}\left[  \left(  X-\mathbb{E}\left[  X\right]  \right)
Y\right]  \right\vert  &  =\left\vert \mathbb{E}\left[  \left(  X-\mathbb{E}%
\left[  X\right]  \right)  \left(  Y-\mathbb{E}\left[  Y\right]  \right)
\right]  \right\vert \nonumber\\
&  =\left\vert \operatorname*{cov}\left(  X,Y\right)  \right\vert
\leq\operatorname*{var}\nolimits^{1/2}\left(  X\right)  \operatorname*{var}%
\nolimits^{1/2}\left(  Y\right)  \text{.}\label{csEV}%
\end{align}
Furthermore, we will be using \cite[Chapter 3]{HJ91B}%
\begin{equation}
\left\vert \operatorname*{tr}\left[  \mathbf{AB}\right]  \right\vert
\leq\left\Vert \mathbf{AB}\right\Vert _{\operatorname*{tr}}\leq\left\Vert
\mathbf{A}\right\Vert _{\operatorname*{tr}}\left\Vert \mathbf{B}\right\Vert
\text{.}\label{singularbound}%
\end{equation}
In particular, if $\mathbf{A}$ is Hermitian nonnegative, we can write
\begin{equation}
\left\vert \operatorname*{tr}\left(  \mathbf{AB}\right)  \right\vert
\leq\left\Vert \mathbf{B}\right\Vert \operatorname*{tr}\left(  \mathbf{A}%
\right)  \text{.}\label{eq_ineq_spectral_norm}%
\end{equation}
Moreover, we will also repeatedly use%
\begin{equation}
\left\Vert \mathbf{AB}\right\Vert _{F}\leq\left\Vert \mathbf{A}\right\Vert
\left\Vert \mathbf{B}\right\Vert _{F}\text{.}\label{inequalityHS}%
\end{equation}

We further provide some inequalities involving the notation and elements
defined in Section \ref{sectionDefAss}. In particular, the following
inequality will be used in the proof of the variance controls given by Lemma
\ref{lemmaVC_Phi_Psi}:%
\begin{equation}
\sup_{M\geq1}\left\Vert \mathbf{Q}\right\Vert ^{p}<+\infty\text{,\quad
a.s.}\label{inequalityQ}%
\end{equation}
Indeed, using the fact that $\left(  \frac{1}{N}\mathbf{XTX}^{\ast}%
+\alpha\mathbf{R}^{-1}\right)  \geq\alpha\mathbf{R}^{-1}$ a.s.%
\footnote{almost surely}
(i.e., the random matrix $\frac{1}{N}\mathbf{XTX}^{\ast}$ is semi-positive
definite with probability one, having $\left\vert M-N\right\vert $ zero
eingevalues), notice that $\left\Vert \mathbf{Q}\right\Vert \leq\alpha
^{-1}\left\Vert \mathbf{R}\right\Vert \leq\alpha^{-1}\left\Vert \mathbf{R}%
\right\Vert _{\sup}$ with probability one.

From the previous inequalities, it also follows that%
\begin{equation}
\sup_{M\geq1}\operatorname*{tr}\left[  \mathbf{\Theta}_{M}\mathbf{Q}_{M}%
^{k}\right]  \leq\alpha^{-k}\left\Vert \mathbf{R}\right\Vert _{\sup}%
^{k}\left\Vert \mathbf{\Theta}\right\Vert _{\sup}<+\infty\text{,\quad
a.s.}\label{supQ}%
\end{equation}

The following two lemmas can be derived as in \cite{HKLNP08}.%

\begin{lemma}%
\label{lemmaBoundsSigmas}%
The quantities $\delta_{M}$, $\tilde{\delta}_{M}$ accept the following upper
and lower bounds:%
\begin{align*}
\delta_{\inf}  &  \leq\delta_{M}\leq c_{\sup}\alpha^{-1}\left\Vert
\mathbf{R}\right\Vert _{\sup}\text{,}\\
\tilde{\delta}_{\inf}  &  \leq\tilde{\delta}_{M}\leq\left\Vert \mathbf{T}%
\right\Vert _{\sup}\text{,}%
\end{align*}
where we have defined
\[
\delta_{\inf}=\frac{c_{\inf}\left\Vert \mathbf{R}\right\Vert _{\inf}}%
{\alpha+\left\Vert \mathbf{R}\right\Vert _{\sup}\left\Vert \mathbf{T}%
\right\Vert _{\sup}}\text{,}\quad\tilde{\delta}_{\inf}=\frac{\alpha\left\Vert
\mathbf{T}\right\Vert _{\inf}}{\alpha+c_{\sup}\left\Vert \mathbf{R}\right\Vert
_{\sup}\left\Vert \mathbf{T}\right\Vert _{\sup}}\text{.}%
\]%
\end{lemma}%

Similarly, observe that,%
\begin{equation}
\gamma_{M}\leq c_{\sup}\alpha^{-2}\left\Vert \mathbf{R}\right\Vert _{\sup}%
^{2}\text{,}\quad\tilde{\gamma}_{M}\leq\left\Vert \mathbf{T}\right\Vert
_{\sup}^{2}\text{.}\label{upperBoundGammas}%
\end{equation}
Additionally, thanks to Jensen's inequality, the lower bounds on the
quantities $\delta_{M}$, $\tilde{\delta}_{M}$ directly imply that
\begin{equation}
\gamma_{M}\geq\frac{1}{c_{\sup}}\delta_{\inf}^{2}>0\text{,}\quad\tilde{\gamma
}_{M}\geq\tilde{\delta}_{\inf}^{2}>0\text{.}\label{lowerBoundGammas}%
\end{equation}
%

\begin{lemma}%
\label{lemmaBounds1-Gammas}%
The quantity $1-\gamma_{M}\tilde{\gamma}_{M}$ accepts the following upper and
lower bounds:%
\[
1-\gamma_{M}\tilde{\gamma}_{M}\leq1-\alpha\frac{\left\Vert \mathbf{R}%
\right\Vert _{\inf}\left\Vert \mathbf{T}\right\Vert _{\inf}}{\left(
\alpha+\left\Vert \mathbf{R}\right\Vert _{\sup}\left\Vert \mathbf{T}%
\right\Vert _{\sup}\right)  \left(  \alpha+c_{\sup}\left\Vert \mathbf{R}%
\right\Vert _{\sup}\left\Vert \mathbf{T}\right\Vert _{\sup}\right)  }<1
\]
and%
\[
1-\gamma_{M}\tilde{\gamma}_{M}\geq\frac{1}{c_{\sup}^{2}}\frac{\alpha^{2}%
}{\left\Vert \mathbf{R}\right\Vert _{\sup}^{2}}\delta_{\inf}^{2}.
\]%
\end{lemma}%

In general, we have, for any finite $k>0$,%

\begin{align}
\sup_{M\geq1}\operatorname*{tr}\left[  \mathbf{\Theta}_{M}\mathbf{E}_{M}%
^{k}\right]   &  \leq\alpha^{-k}\left\Vert \mathbf{R}\right\Vert _{\sup}%
^{k}\left\Vert \mathbf{\Theta}\right\Vert _{\sup}<+\infty\text{,}%
\label{supE}\\
\sup_{N\geq1}\operatorname*{tr}\left[  \mathbf{\tilde{\Theta}}_{N}%
\mathbf{\tilde{E}}_{N}^{k}\right]   &  \leq\left\Vert \mathbf{T}\right\Vert
_{\sup}^{k}\left\Vert \mathbf{\tilde{\Theta}}\right\Vert _{\sup}%
<+\infty\text{,}\label{supEtilde}%
\end{align}
and also%
\begin{align}
\inf_{M\geq1}\operatorname*{tr}\left[  \mathbf{\Theta}_{M}\mathbf{E}_{M}%
^{k}\right]   &  \geq\frac{\theta_{\inf}\alpha^{k}\left\Vert \mathbf{R}%
\right\Vert _{\inf}^{k}}{\left(  \alpha^{2}+c_{\sup}\left\Vert \mathbf{R}%
\right\Vert _{\sup}\left\Vert \mathbf{R}\right\Vert _{\inf}\right)  ^{k}%
}>0\text{,}\label{infE}\\
\inf_{N\geq1}\operatorname*{tr}\left[  \mathbf{\tilde{\Theta}}_{N}%
\mathbf{\tilde{E}}_{N}^{k}\right]   &  \geq\frac{\tilde{\theta}_{\inf}%
\alpha^{k}\left\Vert \mathbf{T}\right\Vert _{\inf}^{k}}{\left(  \alpha
+c_{\sup}\left\Vert \mathbf{R}\right\Vert _{\sup}\left\Vert \mathbf{T}%
\right\Vert _{\sup}\right)  ^{k}}>0\text{.}\label{infEtilde}%
\end{align}

In particular, if $\mathbf{\Theta}_{M}=\mathbf{u}_{M}\mathbf{u}_{M}^{\ast}$,
then $\operatorname*{tr}\left[  \mathbf{\Theta}_{M}\right]  =\left\Vert
\mathbf{u}_{M}\right\Vert ^{2}$, and using the fact that $\left\Vert
\mathbf{s}\right\Vert ^{2}=1$, we notice that $\inf_{M\geq1}\left\Vert
\mathbf{u}_{M}\right\Vert ^{2}\geq\left\Vert \mathbf{R}\right\Vert _{\sup
}^{-1}>0$ and, additionally, $\sup_{M\geq1}\left\Vert \mathbf{u}\right\Vert
^{2}\leq\left\Vert \mathbf{R}\right\Vert _{\inf}^{-1}<+\infty$, so that it
follows from the above inequalities that%
\begin{align}
\max\sup_{M\geq1}\left\{  a_{M},b_{M}\right\}   &  <+\infty\text{,\quad
a.s.,}\label{sup_ab}\\
\max\sup_{M\geq1}\left\{  \bar{a}_{M},\bar{b}_{M}\right\}   &  <+\infty
\text{,}\label{sup_aBar_bBar}\\
\min\inf_{M\geq1}\left\{  \bar{a}_{M},\bar{b}_{M}\right\}   &  >0\text{.}%
\label{inf_aBar_bBar}%
\end{align}
Moreover, observe that for a positive definite matrix $\mathbf{A}\in
\mathbb{C}^{M\times M}$ the Cauchy-Schwarz inequality implies that $\left(
\mathbf{u}^{\ast}\mathbf{Au}\right)  ^{2}\leq\mathbf{u}^{\ast}\mathbf{A}%
^{2}\mathbf{u}$, for all $M$, and hence, using the bounds for $1-\gamma
_{M}\tilde{\gamma}_{M}$ above,
\begin{align}
\sup_{M\geq1}\frac{1}{b_{M}-a_{M}^{2}} &  <+\infty\text{,\quad a.s.,}%
\label{bounds_ratio_b-a2}\\
\sup_{M\geq1}\frac{1}{\bar{b}_{M}-\bar{a}_{M}^{2}} &  <+\infty\text{.}%
\label{bounds_ratio_bBar-aBar2}%
\end{align}

\section{Proof of Lemma \ref{lemmaVC_Phi_Psi} (Variance Controls)%
\label{sectionAppVC}%
}

We first consider the quantities $\Phi_{M}^{\left(  k\right)  }\left(
\mathbf{X}\right)  $, $k\geq1$. Using the Nash-Poincar\'{e} inequality in
(\ref{eq_Nash_Poincare}) and Jensen's inequality, we get by applying
conventional differentiation rules for real-valued functions of complex matrix
arguments (cf. \cite[Section III]{HKLNP08}) along with the chain rule and
after gathering terms together,%
\begin{align*}
\operatorname*{var}\left(  \Phi_{M}^{\left(  k\right)  }\left(  \mathbf{X}%
\right)  \right)   &  \leq k\sum_{r=1}^{k}\frac{1}{N}\mathbb{E}%
\operatorname*{tr}\left[  \mathbf{Q}^{r}\mathbf{\Theta Q}^{2\left(
k-r+1\right)  }\mathbf{\Theta}^{\ast}\mathbf{Q}^{r}\frac{\mathbf{XT}%
^{2}\mathbf{X}^{\ast}}{N}\right] \\
&  +\frac{1}{N}\mathbb{E}\operatorname*{tr}\left[  \mathbf{Q}^{r}%
\mathbf{\Theta}^{\mathbf{\ast}}\mathbf{Q}^{2\left(  k-r+1\right)
}\mathbf{\Theta Q}^{r}\frac{\mathbf{XT}^{2}\mathbf{X}^{\ast}}{N}\right]
\text{,}%
\end{align*}
where we have used%
\[
\frac{\partial\left(  \mathbf{X}\right)  }{\partial X_{ij}}=\mathbf{e}%
_{i}\mathbf{e}_{j}^{T}\text{,\quad}\frac{\partial\left(  \mathbf{X}^{\ast
}\right)  }{\partial\overline{X_{ij}}}=\mathbf{e}_{j}\mathbf{e}_{i}%
^{T}\text{,}%
\]
with $\mathbf{e}_{i}$ being the unit norm vector whose $i$th entry is $1$.
Then, we further notice that, for any two constants $p,q\geq1$,%
\begin{multline*}
\mathbb{E}\operatorname*{tr}\left[  \mathbf{Q}^{p}\mathbf{\Theta Q}%
^{2q}\mathbf{\Theta}^{\ast}\mathbf{Q}^{p}\frac{\mathbf{XT}^{2}\mathbf{X}%
^{\ast}}{N}\right]  \leq\mathbb{E}\left[  \left\Vert \frac{\mathbf{X\mathbf{T}%
}^{2}\mathbf{X}^{\ast}}{N}\right\Vert \operatorname*{tr}\left[  \mathbf{Q}%
^{p}\mathbf{\Theta Q}^{2q}\mathbf{\Theta}^{\ast}\mathbf{Q}^{p}\right]  \right]
\\
=\mathbb{E}\left[  \left\Vert \frac{\mathbf{X\mathbf{T}}^{2}\mathbf{X}^{\ast}%
}{N}\right\Vert \left\Vert \mathbf{\Theta Q}^{p+q}\right\Vert _{F}^{2}\right]
\leq\left\Vert \mathbf{\Theta}\right\Vert _{F}^{2}\mathbb{E}\left[  \left\Vert
\mathbf{Q}\right\Vert ^{2\left(  p+q\right)  }\left\Vert \frac
{\mathbf{X\mathbf{T}}^{2}\mathbf{X}^{\ast}}{N}\right\Vert \right] \\
\leq K\left\Vert \mathbf{\Theta}\right\Vert _{F}^{2}\mathbb{E}\left\Vert
\frac{\mathbf{X\mathbf{T}}^{2}\mathbf{X}^{\ast}}{N}\right\Vert =\mathcal{O}%
\left(  \left\Vert \mathbf{\Theta}\right\Vert _{F}^{2}\right)
\end{multline*}
where we have used inequalities (\ref{inequalityHS}), (\ref{inequalityQ}) and
(\ref{eq_bounded_moments}).

We finally consider the random variables $\Psi_{M}^{\left(  k\right)  }\left(
\mathbf{X}\right)  $, $k\geq1$. By the Nash-Poincar\'{e} and Jensen's
inequality, and similarly as in the previous case, we can write%
\begin{multline*}
\operatorname*{var}\left(  \Psi_{M}^{\left(  k\right)  }\left(  \mathbf{X}%
\right)  \right)  \leq\frac{k+1}{N}\mathbb{E}\operatorname*{tr}\left[
\mathbf{\Theta Q}^{2k}\mathbf{\Theta}^{\ast}\frac{\mathbf{X}\mathbf{\tilde
{Z}^{\ast}\tilde{Z}}\mathbf{X}^{\ast}}{N}\right] \\
+\frac{k+1}{N}\mathbb{E}\operatorname*{tr}\left[  \mathbf{Q}^{k}%
\mathbf{\Theta}^{\ast}\mathbf{\Theta Q}^{k}\frac{\mathbf{X}\mathbf{\tilde{Z}%
}\mathbf{\tilde{Z}}^{\ast}\mathbf{X}^{\ast}}{N}\right] \\
+\left(  k+1\right)  \sum_{r=1}^{k}\frac{1}{N}\mathbb{E}\operatorname*{tr}%
\left[  \mathbf{Q}^{r}\frac{\mathbf{X}\mathbf{\tilde{Z}}\mathbf{X}^{\ast}}%
{N}\mathbf{\Theta Q}^{2\left(  k-r+1\right)  }\mathbf{\Theta}^{\ast}%
\frac{\mathbf{X\tilde{Z}}^{\ast}\mathbf{X}^{\ast}}{N}\mathbf{Q}^{r}%
\frac{\mathbf{XT}^{2}\mathbf{X}^{\ast}}{N}\right] \\
+\frac{1}{N}\mathbb{E}\operatorname*{tr}\left[  \mathbf{Q}^{r}\mathbf{\Theta
}^{\mathbf{\ast}}\frac{\mathbf{X\tilde{Z}}^{\ast}\mathbf{X}^{\ast}}%
{N}\mathbf{Q}^{2\left(  k-r+1\right)  }\frac{\mathbf{X}\mathbf{\tilde{Z}%
}\mathbf{X}^{\ast}}{N}\mathbf{\Theta Q}^{r}\frac{\mathbf{XT}^{2}%
\mathbf{X}^{\ast}}{N}\right]  \text{.}%
\end{multline*}
Then, observe that, for any two constants $p,q\geq1$,%
\begin{multline*}
\mathbb{E}\operatorname*{tr}\left[  \mathbf{Q}^{p}\frac{\mathbf{X}%
\mathbf{\tilde{Z}}\mathbf{X}^{\ast}}{N}\mathbf{\Theta Q}^{2q}\mathbf{\Theta
}^{\ast}\frac{\mathbf{X\tilde{Z}}^{\ast}\mathbf{X}^{\ast}}{N}\mathbf{Q}%
^{p}\frac{\mathbf{XT}^{2}\mathbf{X}^{\ast}}{N}\right]  \leq\\
\leq\mathbb{E}\left[  \left\Vert \frac{\mathbf{X\mathbf{T}}^{2}\mathbf{X}%
^{\ast}}{N}\right\Vert \operatorname*{tr}\left[  \mathbf{Q}^{p}\frac
{\mathbf{X}\mathbf{\tilde{Z}}\mathbf{X}^{\ast}}{N}\mathbf{\Theta Q}%
^{2q}\mathbf{\Theta}^{\ast}\frac{\mathbf{X\tilde{Z}}^{\ast}\mathbf{X}^{\ast}%
}{N}\mathbf{Q}^{p}\right]  \right] \\
\leq\mathbb{E}\left[  \left\Vert \frac{\mathbf{X\mathbf{T}}^{2}\mathbf{X}%
^{\ast}}{N}\right\Vert \left\Vert \frac{\mathbf{X}\mathbf{\tilde{Z}}%
\mathbf{X}^{\ast}}{N}\right\Vert ^{2}\left\Vert \mathbf{Q}\right\Vert
^{2p}\left\Vert \mathbf{\Theta Q}^{q}\right\Vert _{F}^{2}\right] \\
\leq K\left\Vert \mathbf{\Theta}\right\Vert _{F}^{2}\mathbb{E}^{1/2}\left[
\left\Vert \frac{\mathbf{X\mathbf{T}}^{2}\mathbf{X}^{\ast}}{N}\right\Vert
^{2}\right]  \mathbb{E}^{1/2}\left[  \left\Vert \frac{\mathbf{X}%
\mathbf{\tilde{Z}}\mathbf{X}^{\ast}}{N}\right\Vert ^{4}\right]  =\mathcal{O}%
\left(  \left\Vert \mathbf{\Theta}\right\Vert _{F}^{2}\right)  \text{,}%
\end{multline*}
where we have used the Cauchy-Schwarz inequality, along with the inequalities
(\ref{inequalityHS}), (\ref{inequalityQ}) and (\ref{eq_bounded_moments}).

\section{Proof of Propositions \ref{theoremEV1} to \ref{theoremEV4} (Expected
value estimates)%
\label{sectionAppEV}%
}

Let us start by studying the following quantity, namely%
\begin{equation}
\mathbb{E}\left[  \mathbf{Q}^{k}\frac{\mathbf{X\tilde{Z}X}^{\ast}}{N}\right]
_{ij}=\frac{1}{N}\sum_{l=1}^{N}\left[  \mathbf{\tilde{Z}}\right]
_{l}\mathbb{E}\left[  \mathbf{Q}^{k}\mathbf{x}_{l}\mathbf{x}_{l}^{\ast
}\right]  _{ij}\text{.}\label{1k}%
\end{equation}
Using the integration by parts formula in (\ref{eq_integration_by_parts}), we
find that ($t_{l}=\left[  \mathbf{T}\right]  _{l}$)%
\begin{align}
\mathbb{E}\left[  \mathbf{Q}^{k}\mathbf{x}_{l}\mathbf{x}_{l}^{\ast}\right]
_{ij} &  =\sum_{r=1}^{M}\mathbb{E}\left[  \left[  \mathbf{Q}^{k}\right]
_{ir}X_{rl}\overline{X_{jl}}\right] \nonumber\\
&  =\frac{\mathbb{E}\left[  \mathbf{Q}^{k}\right]  _{ij}}{1+t_{l}\frac{1}%
{N}\mathbb{E}\operatorname*{tr}\left[  \mathbf{Q}\right]  }-\frac{t_{l}%
}{1+t_{l}\frac{1}{N}\mathbb{E}\operatorname*{tr}\left[  \mathbf{Q}\right]
}\mathbb{E}\left[  \chi_{M}^{\left(  1\right)  }\left[  \mathbf{Q}%
^{k}\mathbf{x}_{l}\mathbf{x}_{l}^{\ast}\right]  _{ij}\right] \nonumber\\
&  -\sum_{p=1}^{k-1}\frac{t_{l}}{1+t_{l}\frac{1}{N}\mathbb{E}%
\operatorname*{tr}\left[  \mathbf{Q}\right]  }\mathbb{E}\left[  \left[
\mathbf{Q}^{p}\mathbf{x}_{l}\mathbf{x}_{l}^{\ast}\right]  _{ij}\frac{1}%
{N}\operatorname*{tr}\left[  \mathbf{Q}^{k-p+1}\right]  \right]
\text{,}\label{2k}%
\end{align}
By plugging (\ref{2k}) into (\ref{1k}), we obtain%
\begin{multline}
\mathbb{E}\left[  \mathbf{Q}^{k}\frac{\mathbf{X}\mathbf{\tilde{Z}}%
\mathbf{X}^{\ast}}{N}\right]  _{ij}=\frac{1}{N}\operatorname*{tr}\left[
\mathbf{\tilde{Z}}\left(  \mathbf{I}_{N}+\frac{1}{N}\left(  \mathbb{E}%
\operatorname*{tr}\left[  \mathbf{Q}\right]  \right)  \mathbf{T}\right)
^{-1}\right]  \mathbb{E}\left[  \mathbf{Q}^{k}\right]  _{ij}\\
-\sum_{p=1}^{k-1}\mathbb{E}\left[  \mathbf{Q}^{p}\frac{\mathbf{X}%
\mathbf{\tilde{Z}\tilde{F}X}^{\ast}}{N}\right]  _{ij}\frac{1}{N}%
\mathbb{E}\operatorname*{tr}\left[  \mathbf{Q}^{k-p+1}\right]  -\sum_{q=1}%
^{k}\mathbb{E}\left[  \chi_{M}^{\left(  q\right)  }\left[  \mathbf{Q}%
^{k-q+1}\frac{\mathbf{X}\mathbf{\tilde{Z}\tilde{F}X}^{\ast}}{N}\right]
_{ij}\right]  \text{.}\label{3k}%
\end{multline}
Furthermore, from the expression in (\ref{RI}), we observe that%
\begin{equation}
\mathbb{E}\left[  \mathbf{Q}^{k}\right]  _{ij}=\alpha^{-1}\mathbb{E}\left[
\mathbf{Q}^{k-1}\right]  _{ij}\left[  \mathbf{R}\right]  _{j}-\alpha
^{-1}\mathbb{E}\left[  \mathbf{Q}^{k}\frac{\mathbf{XTX}^{\ast}}{N}\right]
_{ij}\left[  \mathbf{R}\right]  _{j}\text{.}\label{RIk}%
\end{equation}
Then, by using in (\ref{RIk}) the identity (\ref{3k}) with $\mathbf{\tilde{Z}%
}=\mathbf{T}$ along with the definition of the matrix $\mathbf{E}$, i.e.,
$\left[  \mathbf{E}\right]  _{j}=\left[  \left(  \tilde{\delta}_{M}%
\mathbf{I}_{M}+\alpha\mathbf{R}^{-1}\right)  ^{-1}\right]  _{j}$, after some
algebraic manipulations we get the following expression:%
\begin{multline}
\mathbb{E}\left[  \mathbf{Q}^{k}\right]  _{ij}=\mathbb{E}\left[
\mathbf{Q}^{k-1}\mathbf{E}\right]  _{ij}+\left(  \tilde{\delta}_{M}-\frac
{1}{N}\operatorname*{tr}\left[  \mathbf{\tilde{F}}\right]  \right)
\mathbb{E}\left[  \mathbf{Q}^{k}\mathbf{E}\right]  _{ij}\\
+\sum_{p=1}^{k-1}\mathbb{E}\left[  \mathbf{Q}^{p}\frac{\mathbf{X}%
\mathbf{\tilde{Z}\tilde{F}X}^{\ast}}{N}\mathbf{E}\right]  _{ij}\frac{1}%
{N}\mathbb{E}\operatorname*{tr}\left[  \mathbf{Q}^{k-p+1}\right]  +\sum
_{q=1}^{k}\mathbb{E}\left[  \chi_{M}^{\left(  q\right)  }\left[
\mathbf{Q}^{k-q+1}\frac{\mathbf{X}\mathbf{\tilde{Z}\tilde{F}X}^{\ast}}%
{N}\mathbf{E}\right]  _{ij}\right]  \text{.}\label{4k}%
\end{multline}

In particular, from the expressions in (\ref{3k}) and (\ref{4k}), we obtain%
\begin{align}
\frac{1}{N}\mathbb{E}\operatorname*{tr}\left[  \mathbf{Q}^{k}\right]   &
=\frac{1}{N}\mathbb{E}\operatorname*{tr}\left[  \mathbf{EQ}^{k-1}\right]
+\left(  \tilde{\delta}_{M}-\frac{1}{N}\operatorname*{tr}\left[
\mathbf{\tilde{F}}\right]  \right)  \frac{1}{N}\mathbb{E}\operatorname*{tr}%
\left[  \mathbf{EQ}^{k}\right] \nonumber\\
&  +\sum_{p=1}^{k-1}\frac{1}{N}\mathbb{E}\operatorname*{tr}\left[
\mathbf{EQ}^{p}\frac{\mathbf{XT\tilde{F}X}^{\ast}}{N}\right]  \frac{1}%
{N}\mathbb{E}\operatorname*{tr}\left[  \mathbf{Q}^{k-p+1}\right]  +\sum
_{q=1}^{k}\mathcal{X}_{M}^{\left(  q\right)  }\text{,}\label{5k}%
\end{align}
where we have defined the following error terms (here $\mathbf{\Theta}%
=\frac{1}{N}\mathbf{I}_{M}$, $\mathbf{Z}=\mathbf{E}$ and $\mathbf{\tilde{Z}%
}=\mathbf{T}$):%
\[
\mathcal{X}_{M}^{\left(  q\right)  }=\mathbb{E}\left[  \chi_{M}^{\left(
q\right)  }\operatorname*{tr}\left[  \mathbf{Z\Theta Q}^{k-q+1}\frac
{\mathbf{X\tilde{Z}\tilde{F}X}^{\ast}}{N}\right]  \right]  \text{.}%
\]
Before proceeding further, notice that (\ref{csEV}) along with Lemma
\ref{lemmaVC_Phi_Psi} implies that, for any $q\geq1$,%
\begin{equation}
\mathcal{X}_{M}^{\left(  q\right)  }=\mathcal{O}\left(  \frac{\left\Vert
\mathbf{\Theta}\right\Vert _{F}}{N^{3/2}}\right)  \text{.}\label{xiMq}%
\end{equation}

We now elaborate on (\ref{5k}) in the case $k=1$. Specifically, note that we
can write%
\begin{multline*}
\frac{1}{N}\mathbb{E}\operatorname*{tr}\left[  \mathbf{Q}\right]  -\frac{1}%
{N}\mathbb{E}\operatorname*{tr}\left[  \mathbf{E}\right]  =\left(
\tilde{\delta}_{M}-\frac{1}{N}\operatorname*{tr}\left[  \mathbf{\tilde{F}%
}\right]  \right)  \frac{1}{N}\mathbb{E}\operatorname*{tr}\left[
\mathbf{EQ}\right]  +\mathcal{O}\left(  N^{-2}\right) \\
=\left(  \frac{1}{N}\mathbb{E}\operatorname*{tr}\left[  \mathbf{Q}\right]
-\frac{1}{N}\operatorname*{tr}\left[  \mathbf{E}\right]  \right)
\operatorname*{tr}\left[  \mathbf{\tilde{F}\tilde{E}}\right]  \frac{1}%
{N}\mathbb{E}\operatorname*{tr}\left[  \mathbf{EQ}\right]  +\mathcal{O}\left(
N^{-2}\right)  \text{,}%
\end{multline*}
and so we get%
\[
\left(  \frac{1}{N}\mathbb{E}\operatorname*{tr}\left[  \mathbf{Q}\right]
-\frac{1}{N}\operatorname*{tr}\left[  \mathbf{E}\right]  \right)  \left(
1-\frac{1}{N}\operatorname*{tr}\left[  \mathbf{\tilde{E}\tilde{F}}\right]
\frac{1}{N}\mathbb{E}\operatorname*{tr}\left[  \mathbf{EQ}\right]  \right)
=\mathcal{O}\left(  N^{-2}\right)  \text{.}%
\]
Moreover, using (\ref{eq_ineq_spectral_norm}), we observe that, uniformly in
$M$,%
\begin{equation}
\left\vert \frac{1}{N}\operatorname*{tr}\left[  \mathbf{\tilde{E}\tilde{F}%
}\right]  \frac{1}{N}\mathbb{E}\operatorname*{tr}\left[  \mathbf{EQ}\right]
\right\vert \leq\tilde{\delta}_{M}\frac{1}{N}\mathbb{E}\operatorname*{tr}%
\left[  \mathbf{Q}\right]  \left\Vert \mathbf{\tilde{F}}\right\Vert \left\Vert
\mathbf{E}\right\Vert <1\text{,}\nonumber
\end{equation}
which follows by Assumption \textbf{(As2)} from the fact that%
\[
\sup_{M\geq1}\max\left\{  \left\Vert \left(  \mathbf{I}_{M}+\alpha
\tilde{\delta}_{M}^{-1}\mathbf{R}^{-1}\right)  ^{-1}\right\Vert ,\left[
\left(  \mathbf{I}_{N}+\left(  \frac{1}{N}\mathbb{E}\operatorname*{tr}\left[
\mathbf{Q}\right]  \mathbf{T}\right)  ^{-1}\right)  ^{-1}\right]  \right\}
<1\text{.}%
\]
Hence, we have%
\[
\frac{1}{N}\mathbb{E}\operatorname*{tr}\left[  \mathbf{Q}\right]  =\frac{1}%
{N}\operatorname*{tr}\left[  \mathbf{E}\right]  +\mathcal{O}\left(  \frac
{1}{N^{2}}\right)  \text{.}%
\]
In particular, noting that%
\[
\operatorname*{tr}\left[  \mathbf{\Theta}\left(  \mathbf{\tilde{E}%
}-\mathbf{\tilde{F}}\right)  \right]  =\left(  \frac{1}{N}\mathbb{E}%
\operatorname*{tr}\left[  \mathbf{Q}\right]  -\frac{1}{N}\operatorname*{tr}%
\left[  \mathbf{E}\right]  \right)  \operatorname*{tr}\left[  \mathbf{\Theta
\tilde{F}\tilde{E}}\right]  \text{,}%
\]
together with $\left\Vert \mathbf{\Theta}\right\Vert _{\operatorname*{tr}}%
\leq\sqrt{M}\left\Vert \mathbf{\Theta}\right\Vert _{F}$ and $\sup_{M\geq
1}\left\Vert \mathbf{\tilde{F}}\right\Vert \leq\left\Vert \mathbf{T}%
\right\Vert _{\sup}$, the next result follows straightforwardly.%

\begin{lemma}%
\label{OFtEt}%
With all above definitions, the following approximation rule holds:%
\begin{equation}
\operatorname*{tr}\left[  \mathbf{\Theta\tilde{F}}\right]  =\operatorname*{tr}%
\left[  \mathbf{\Theta\tilde{E}}\right]  +\mathcal{O}\left(  \frac{\left\Vert
\mathbf{\Theta}\right\Vert _{F}}{N^{3/2}}\right)  \text{.}\label{OFtEt}%
\end{equation}%
\end{lemma}%

The variance control in (\ref{xiMq}) along with (\ref{OFtEt}) imply that%

\begin{multline}
\frac{1}{N}\mathbb{E}\operatorname*{tr}\left[  \mathbf{Q}^{k}\right]
=\frac{1}{N}\mathbb{E}\operatorname*{tr}\left[  \mathbf{EQ}^{k-1}\right]
+\label{Atk}\\
+\sum_{p=1}^{k-1}\frac{1}{N}\mathbb{E}\operatorname*{tr}\left[  \mathbf{EQ}%
^{p}\frac{\mathbf{XT\tilde{F}X}^{\ast}}{N}\right]  \frac{1}{N}\mathbb{E}%
\operatorname*{tr}\left[  \mathbf{Q}^{k-p+1}\right]  +\mathcal{O}\left(
\frac{1}{N^{2}}\right)  \text{.}%
\end{multline}
Similarly, we can write the following estimates from (\ref{4k}) and
(\ref{3k}), respectively,%
\begin{align}
\mathbb{E}\left[  \mathbf{\Theta Q}^{k}\frac{\mathbf{X}\mathbf{\tilde{Z}%
}\mathbf{X}^{\ast}}{N}\right]   &  =\frac{1}{N}\operatorname*{tr}\left[
\mathbf{\tilde{Z}}\left(  \mathbf{I}_{N}+\frac{1}{N}\left(  \mathbb{E}%
\operatorname*{tr}\left[  \mathbf{Q}\right]  \right)  \mathbf{T}\right)
^{-1}\right]  \mathbb{E}\left[  \mathbf{\Theta Q}^{k}\right] \nonumber\\
&  -\sum_{p=1}^{k-1}\mathbb{E}\left[  \mathbf{\Theta Q}^{p}\frac
{\mathbf{X}\mathbf{\tilde{Z}\tilde{F}X}^{\ast}}{N}\right]  \frac{1}%
{N}\mathbb{E}\operatorname*{tr}\left[  \mathbf{Q}^{k-p+1}\right]
+\mathcal{O}\left(  \frac{\left\Vert \mathbf{\Theta}\right\Vert _{F}}{N^{3/2}%
}\right)  \text{,}\label{Btk}%
\end{align}
and%
\begin{multline}
\mathbb{E}\left[  \mathbf{\Theta Q}^{k}\right]  =\mathbb{E}\left[
\mathbf{\Theta Q}^{k-1}\mathbf{E}\right]  +\label{Ctk}\\
+\sum_{p=1}^{k-1}\mathbb{E}\left[  \mathbf{\Theta Q}^{p}\frac{\mathbf{XT\tilde
{F}X}^{\ast}}{N}\mathbf{E}\right]  \frac{1}{N}\mathbb{E}\operatorname*{tr}%
\left[  \mathbf{Q}^{k-p+1}\right]  +\mathcal{O}\left(  \frac{\left\Vert
\mathbf{\Theta}\right\Vert _{F}}{N^{3/2}}\right)  \text{.}%
\end{multline}

Now, the proof of Propositions \ref{theoremEV1} to \ref{theoremEV4} can be now
readily completed by handling the estimates (\ref{Atk}), (\ref{Btk}) and
(\ref{Ctk}), successively, following an iterative scheme from $k=1$ to $k=4$.

\section{Proof of Proposition \ref{proposition_aM_bM}%
\label{AppProposition_aM_bM}%
}

We concentrate first on (\ref{prop_aM}). Observing that%
\[
\mathbb{E}\left[  \left[  \mathbf{Q}\right]  _{ij}\Psi\left(  \omega\right)
\right]  =\alpha^{-1}\left[  \mathbf{R}\right]  _{j}\mathbb{E}\left[
\Psi\left(  \omega\right)  \right]  -\alpha^{-1}\mathbb{E}\left[  \left[
\mathbf{Q}\frac{1}{N}\mathbf{XTX}^{\ast}\mathbf{R}\right]  _{ij}\Psi\left(
\omega\right)  \right]
\]
it is sufficient to investigate the term $\mathbb{E}\left[  \mathbf{Q}\frac
{1}{N}\mathbf{XTX}^{\ast}\Psi\left(  \omega\right)  \right]  _{ij}$. Now,
observe that we can express
\[
\mathbb{E}\left[  \mathbf{Q}\frac{1}{N}\mathbf{XTX}^{\ast}\Psi\left(
\omega\right)  \right]  _{ij}=\frac{1}{N}\sum_{l=1}^{N}t_{l}\mathbb{E}\left[
\mathbf{Qx}_{l}\mathbf{x}_{l}^{\ast}\Psi\left(  \omega\right)  \right]  _{ij}%
\]
and therefore, using the integration by parts formula,\textbf{\ }we get after
some algebraic manipulations%
\begin{align}
\mathbb{E}\left[  \left[  \mathbf{Q}\right]  _{ij}\Psi\left(  \omega\right)
\right]   &  =\left[  \mathbf{E}\right]  _{j}\mathbb{E}\left[  \Psi\left(
\omega\right)  \right] \nonumber\\
&  +\operatorname*{i}\omega A\frac{1}{\sqrt{N}}\mathbb{E}\left[  \left[
\mathbf{Q}^{2}\mathbf{uu}^{\ast}\mathbf{Q}\frac{\mathbf{XT\tilde{E}X}^{\ast}%
}{N}\mathbf{E}\right]  _{ij}\Psi\left(  \omega\right)  \right] \nonumber\\
&  +\operatorname*{i}\omega B\frac{1}{\sqrt{N}}\mathbb{E}\left[  \left[
\mathbf{Q}^{3}\mathbf{uu}^{\ast}\mathbf{Q}\frac{\mathbf{X}\mathbf{T\tilde{E}%
}\mathbf{X}^{\ast}}{N}\mathbf{E}\right]  _{ij}\Psi\left(  \omega\right)
\right] \nonumber\\
&  +\operatorname*{i}\omega B\frac{1}{\sqrt{N}}\mathbb{E}\left[  \left[
\mathbf{Q}^{2}\mathbf{uu}^{\ast}\mathbf{Q}^{2}\frac{\mathbf{X}\mathbf{T\tilde
{E}}\mathbf{X}^{\ast}}{N}\mathbf{E}\right]  _{ij}\Psi\left(  \omega\right)
\right] \nonumber\\
&  +\mathbb{E}\left[  \left(  \frac{1}{N}\operatorname*{tr}\left[
\mathbf{Q}\right]  -\delta_{M}\right)  \left[  \mathbf{Q}\frac
{\mathbf{XT\tilde{E}X}^{\ast}}{N}\mathbf{E}\right]  _{ij}\Psi\left(
\omega\right)  \right]  \text{.}\label{eq_QijPsi}%
\end{align}
Therefore, we can conclude that
\begin{multline*}
\mathbb{E}\left[  \left(  a_{M}-\overline{a}_{M}\right)  \Psi_{M}\left(
\omega\right)  \right]  =\operatorname*{i}\omega A\frac{1}{\sqrt{N}}%
\mathbb{E}\left[  \mathbf{u}^{\ast}\mathbf{Q}^{2}\mathbf{uu}^{\ast}%
\mathbf{Q}\frac{\mathbf{X}\mathbf{T\tilde{E}}\mathbf{X}^{\ast}}{N}%
\mathbf{Eu}\Psi\left(  \omega\right)  \right] \\
+\operatorname*{i}\omega B\frac{1}{\sqrt{N}}\mathbb{E}\left[  \mathbf{u}%
^{\ast}\mathbf{Q}^{3}\mathbf{uu}^{\ast}\mathbf{Q}\frac{\mathbf{X}%
\mathbf{T\tilde{E}}\mathbf{X}^{\ast}}{N}\mathbf{Eu}\Psi\left(  \omega\right)
\right] \\
+\operatorname*{i}\omega B\frac{1}{\sqrt{N}}\mathbb{E}\left[  \mathbf{u}%
^{\ast}\mathbf{Q}^{2}\mathbf{uu}^{\ast}\mathbf{Q}^{2}\frac{\mathbf{X}%
\mathbf{T\tilde{E}}\mathbf{X}^{\ast}}{N}\mathbf{Eu}\Psi\left(  \omega\right)
\right]  +\mathcal{Y}_{1,M}\text{,}%
\end{multline*}
where we have defined
\[
\mathcal{Y}_{1,M}=\mathbb{E}\left[  \left(  \frac{1}{N}\operatorname*{tr}%
\left[  \mathbf{Q}\right]  -\delta_{M}\right)  \mathbf{u}^{\ast}%
\mathbf{Q}\frac{\mathbf{X}\mathbf{T\tilde{E}}\mathbf{X}^{\ast}}{N}%
\mathbf{Eu}\Psi\left(  \omega\right)  \right]  \text{.}%
\]
Hence, after some algebraic manipulations and the application of the variance
controls in Lemma \ref{lemmaVC_Phi_Psi}, we finally obtain
\begin{multline*}
\sqrt{N}\mathbb{E}\left[  \left(  a_{M}-\overline{a}_{M}\right)  \Psi
_{M}\left(  \omega\right)  \right]  =\operatorname*{i}\omega A\mathbb{E}%
\left[  \mathbf{u}^{\ast}\mathbf{Q}^{2}\mathbf{u}\right]  \mathbb{E}\left[
\mathbf{u}^{\ast}\mathbf{Q}\frac{\mathbf{X}\mathbf{T\tilde{E}}\mathbf{X}%
^{\ast}}{N}\mathbf{Eu}\right]  \mathbb{E}\left[  \Psi\left(  \omega\right)
\right] \\
+\operatorname*{i}\omega B\mathbb{E}\left[  \mathbf{u}^{\ast}\mathbf{Q}%
^{3}\mathbf{u}\right]  \mathbb{E}\left[  \mathbf{u}^{\ast}\mathbf{Q}%
\frac{\mathbf{X}\mathbf{T\tilde{E}}\mathbf{X}^{\ast}}{N}\mathbf{Eu}\right]
\mathbb{E}\left[  \Psi\left(  \omega\right)  \right] \\
+\operatorname*{i}\omega B\mathbb{E}\left[  \mathbf{u}^{\ast}\mathbf{Q}%
^{2}\mathbf{u}\right]  \mathbb{E}\left[  \mathbf{u}^{\ast}\mathbf{Q}^{2}%
\frac{\mathbf{X}\mathbf{T\tilde{E}}\mathbf{X}^{\ast}}{N}\mathbf{Eu}\right]
\mathbb{E}\left[  \Psi\left(  \omega\right)  \right]  +\mathcal{O}\left(
N^{-1/2}\right)  \text{,}%
\end{multline*}
and (\ref{prop_aM}) follows by Propositions \ref{theoremEV1} to
\ref{theoremEV3}.

We now deal with (\ref{prop_bM}). Observing that%
\begin{equation}
\mathbb{E}\left[  \left[  \mathbf{Q}^{2}\right]  _{ij}\Psi\left(
\omega\right)  \right]  =\alpha^{-1}\mathbb{E}\left[  \left[  \mathbf{Q}%
\right]  _{ij}\Psi\left(  \omega\right)  \right]  \left[  \mathbf{R}\right]
_{j}-\alpha^{-1}\mathbb{E}\left[  \left[  \mathbf{Q}^{2}\frac{1}%
{N}\mathbf{XTX}^{\ast}\right]  _{ij}\Psi\left(  \omega\right)  \right]
\left[  \mathbf{R}\right]  _{j}\text{,}\label{bMobs}%
\end{equation}
we only need to investigate the quantity
\[
\mathbb{E}\left[  \left[  \mathbf{Q}^{2}\frac{1}{N}\mathbf{XTX}^{\ast}\right]
_{ij}\Psi\left(  \omega\right)  \right]  =\frac{1}{N}\sum_{l=1}^{N}%
t_{l}\mathbb{E}\left[  \left[  \mathbf{Q}^{2}\mathbf{x}_{l}\mathbf{x}%
_{l}^{\ast}\right]  _{ij}\Psi\left(  \omega\right)  \right]  \text{.}%
\]
Thus, we can develop $\mathbb{E}\left[  \left[  \mathbf{Q}^{2}\mathbf{x}%
_{l}\mathbf{x}_{l}^{\ast}\right]  _{ij}\Psi\left(  \omega\right)  \right]  $
by using the integration by parts formula and applying similar algebraic
manipulations as in the proof of (\ref{prop_aM}).\ Then, using the previous
estimate in (\ref{bMobs}), we can write%
\begin{multline*}
\mathbb{E}\left[  \left[  \mathbf{Q}^{2}\right]  _{ij}\Psi\left(
\omega\right)  \right]  =\mathbb{E}\left[  \left[  \mathbf{QE}\right]
_{ij}\Psi\left(  \omega\right)  \right]  +\frac{\tilde{\gamma}}{1-\gamma
\tilde{\gamma}}\mathbb{E}\left[  \left[  \mathbf{Q}\frac{\mathbf{XT\tilde{E}%
X}^{\ast}}{N}\mathbf{E}\right]  _{ij}\Psi\left(  \omega\right)  \right] \\
+\operatorname*{i}\omega A\frac{1}{\sqrt{N}}\mathbb{E}\left[  \left[
\mathbf{Q}^{3}\mathbf{uu}^{\ast}\mathbf{Q}\frac{\mathbf{XT\tilde{E}X}^{\ast}%
}{N}\mathbf{E}\right]  _{ij}\Psi\left(  \omega\right)  \right] \\
+\operatorname*{i}\omega B\frac{1}{\sqrt{N}}\mathbb{E}\left[  \left[
\mathbf{Q}^{4}\mathbf{uu}^{\ast}\mathbf{Q}\frac{\mathbf{XT\tilde{E}X}^{\ast}%
}{N}\mathbf{E}\right]  _{ij}\Psi\left(  \omega\right)  \right] \\
+\operatorname*{i}\omega B\frac{1}{\sqrt{N}}\mathbb{E}\left[  \left[
\mathbf{Q}^{3}\mathbf{uu}^{\ast}\mathbf{Q}^{2}\frac{\mathbf{XT\tilde{E}%
X}^{\ast}}{N}\mathbf{E}\right]  _{ij}\Psi\left(  \omega\right)  \right] \\
+\mathbb{E}\left[  \left(  \frac{1}{N}\operatorname*{tr}\left[  \mathbf{Q}%
\right]  -\delta_{M}\right)  \left[  \mathbf{Q}^{2}\frac{\mathbf{XT\tilde{E}%
X}^{\ast}}{N}\mathbf{E}\right]  _{ij}\Psi\left(  \omega\right)  \right] \\
+\mathbb{E}\left[  \left(  \frac{1}{N}\operatorname*{tr}\left[  \mathbf{Q}%
^{2}\right]  -\frac{\tilde{\gamma}}{1-\gamma\tilde{\gamma}}\right)  \left[
\mathbf{Q}\frac{\mathbf{XT\tilde{E}X}^{\ast}}{N}\mathbf{E}\right]  _{ij}%
\Psi\left(  \omega\right)  \right]  \text{.}%
\end{multline*}
Consequently, we can finally state that
\begin{multline*}
\mathbb{E}\left[  \left(  b_{M}-\overline{b}_{M}\right)  \Psi_{M}\left(
\omega\right)  \right]  =\operatorname*{i}\omega A\frac{1}{\sqrt{N}}%
\mathbb{E}\left[  \mathbf{u}^{\ast}\mathbf{Q}^{3}\mathbf{uu}^{\ast}%
\mathbf{Q}\frac{\mathbf{XT\tilde{E}X}^{\ast}}{N}\mathbf{Eu}\Psi\left(
\omega\right)  \right] \\
+\operatorname*{i}\omega B\frac{1}{\sqrt{N}}\mathbb{E}\left[  \mathbf{u}%
^{\ast}\mathbf{Q}^{4}\mathbf{uu}^{\ast}\mathbf{Q}\frac{\mathbf{XT\tilde{E}%
X}^{\ast}}{N}\mathbf{Eu}\Psi\left(  \omega\right)  \right] \\
+\operatorname*{i}\omega B\frac{1}{\sqrt{N}}\mathbb{E}\left[  \mathbf{u}%
^{\ast}\mathbf{Q}^{3}\mathbf{uu}^{\ast}\mathbf{Q}^{2}\frac{\mathbf{XT\tilde
{E}X}^{\ast}}{N}\mathbf{Eu}\Psi\left(  \omega\right)  \right] \\
+\frac{\gamma}{1-\gamma\tilde{\gamma}}\mathbb{E}\left[  \left(  \mathbf{u}%
^{\ast}\mathbf{Q}\frac{\mathbf{XT\tilde{E}X}^{\ast}}{N}\mathbf{Eu}%
-\tilde{\gamma}\mathbf{u}^{\ast}\mathbf{E}^{2}\mathbf{u}\right)  \Psi\left(
\omega\right)  \right] \\
+\mathbb{E}\left[  \left(  \mathbf{u}^{\ast}\mathbf{QEu}-\mathbf{u}^{\ast
}\mathbf{E}^{2}\mathbf{u}\right)  \Psi\left(  \omega\right)  \right]
+\mathcal{Y}_{2,M}+\mathcal{Y}_{3,M}\text{,}%
\end{multline*}
where we have defined
\begin{align*}
\mathcal{Y}_{2,M} &  =\mathbb{E}\left[  \left(  \frac{1}{N}\operatorname*{tr}%
\left[  \mathbf{Q}\right]  -\delta_{M}\right)  \left[  \mathbf{u}^{\ast
}\mathbf{Q}^{2}\frac{\mathbf{XT\tilde{E}X}^{\ast}}{N}\mathbf{Eu}\right]
\Psi\left(  \omega\right)  \right] \\
\mathcal{Y}_{3,M} &  =\mathbb{E}\left[  \left(  \frac{1}{N}\operatorname*{tr}%
\left[  \mathbf{Q}^{2}\right]  -\frac{\tilde{\gamma}}{1-\gamma\tilde{\gamma}%
}\right)  \left[  \mathbf{u}^{\ast}\mathbf{Q}\frac{\mathbf{XT\tilde{E}X}%
^{\ast}}{N}\mathbf{Eu}\right]  \Psi\left(  \omega\right)  \right]  \text{.}%
\end{align*}
In particular, note that (\ref{csEV}) along with Lemma \ref{lemmaVC_Phi_Psi}
implies that $\mathcal{Y}_{2,M}+\mathcal{Y}_{3,M}=\mathcal{O}(N^{-1})$. On the
other hand, we also notice that%
\begin{multline*}
\mathbb{E}\left[  \left(  \mathbf{u}^{\ast}\mathbf{QEu-\mathbf{u}^{\ast}E}%
^{2}\mathbf{u}\right)  \Psi\left(  \omega\right)  \right]  =\operatorname*{i}%
\omega A\frac{1}{\sqrt{N}}\mathbb{E}\left[  \mathbf{u}^{\ast}\mathbf{Q}%
^{2}\mathbf{uu}^{\ast}\mathbf{Q}\frac{\mathbf{X}\mathbf{T\tilde{E}}%
\mathbf{X}^{\ast}}{N}\mathbf{E}^{2}\mathbf{u}\Psi\left(  \omega\right)
\right] \\
+\operatorname*{i}\omega B\frac{1}{\sqrt{N}}\mathbb{E}\left[  \mathbf{u}%
^{\ast}\mathbf{Q}^{3}\mathbf{uu}^{\ast}\mathbf{Q}\frac{\mathbf{X}%
\mathbf{T\tilde{E}}\mathbf{X}^{\ast}}{N}\mathbf{E}^{2}\mathbf{u}\Psi\left(
\omega\right)  \right] \\
+\operatorname*{i}\omega B\frac{1}{\sqrt{N}}\mathbb{E}\left[  \mathbf{u}%
^{\ast}\mathbf{Q}^{2}\mathbf{uu}^{\ast}\mathbf{Q}^{2}\frac{\mathbf{X}%
\mathbf{T\tilde{E}}\mathbf{X}^{\ast}}{N}\mathbf{E}^{2}\mathbf{u}\Psi\left(
\omega\right)  \right]  +\mathcal{Y}_{4,M}%
\end{multline*}
where
\[
\mathcal{Y}_{4,M}=\mathbb{E}\left[  \left(  \frac{1}{N}\operatorname*{tr}%
\left[  \mathbf{Q}\right]  -\delta_{M}\right)  \mathbf{u}^{\ast}%
\mathbf{Q}\frac{\mathbf{X}\mathbf{T\tilde{E}}\mathbf{X}^{\ast}}{N}%
\mathbf{E}^{2}\mathbf{u}\Psi\left(  \omega\right)  \right]  \text{.}%
\]
It can be trivially seen that $\mathcal{Y}_{4,M}=\mathcal{O}(N^{-1})$.
Furthermore, we readily see that%
\begin{multline*}
\mathbb{E}\left[  \left(  \mathbf{u}^{\ast}\mathbf{Q}\frac{\mathbf{XT\tilde
{E}X}^{\ast}}{N}\mathbf{Eu}-\tilde{\gamma}\mathbf{u}^{\ast}\mathbf{E}%
^{2}\mathbf{u}\right)  \Psi\left(  \omega\right)  \right]  =-\operatorname*{i}%
\omega A\frac{1}{\sqrt{N}}\mathbb{E}\left[  \mathbf{u}^{\ast}\mathbf{Q}%
^{2}\mathbf{uu}^{\ast}\mathbf{Q}\frac{\mathbf{X\mathbf{T}\tilde{E}}_{N}%
^{2}\mathbf{X}^{\ast}}{N}\mathbf{Eu}\Psi\left(  \omega\right)  \right] \\
-\operatorname*{i}\omega B\frac{1}{\sqrt{N}}\mathbb{E}\left[  \mathbf{u}%
^{\ast}\mathbf{Q}^{3}\mathbf{uu}^{\ast}\mathbf{Q}\frac
{\mathbf{X\mathbf{T\tilde{E}}}^{2}\mathbf{X}^{\ast}}{N}\mathbf{Eu}\Psi\left(
\omega\right)  \right] \\
-\operatorname*{i}\omega B\frac{1}{\sqrt{N}}\mathbb{E}\left[  \mathbf{u}%
^{\ast}\mathbf{Q}^{2}\mathbf{uu}^{\ast}\mathbf{Q}^{2}\frac{\mathbf{X\mathbf{T}%
\tilde{E}}_{N}^{2}\mathbf{X}^{\ast}}{N}\mathbf{Eu}\Psi\left(  \omega\right)
\right] \\
+\tilde{\gamma}\mathbb{E}\left[  \left(  \mathbf{u^{\ast}QEu-u^{\ast}E}%
^{2}\mathbf{u}\right)  \Psi\left(  \omega\right)  \right]  +\mathcal{Y}%
_{5,M}\text{,}%
\end{multline*}
where the term $\mathbb{E}\left[  \left(  \mathbf{u}^{\ast}%
\mathbf{QEu-\mathbf{u}^{\ast}E}^{2}\mathbf{u}\right)  \Psi\left(
\omega\right)  \right]  $ has been examined above, and where
\[
\mathcal{Y}_{5,M}=-\mathbb{E}\left[  \left(  \frac{1}{N}\operatorname*{tr}%
\left[  \mathbf{Q}\right]  -\delta_{M}\right)  \mathbf{u}^{\ast}%
\mathbf{Q}\frac{\mathbf{XT\tilde{E}}^{2}\mathbf{X}^{\ast}}{N}\mathbf{Eu}%
\Psi\left(  \omega\right)  \right]  \text{.}%
\]
It can be readily seen that $\mathcal{Y}_{5,M}=\mathcal{O}(N^{-1})$. Inserting
the above back into the original expression, we observe that%
\begin{multline*}
\sqrt{N}\mathbb{E}\left[  \left(  b_{M}-\overline{b}_{M}\right)  \Psi
_{M}\left(  \omega\right)  \right]  =\operatorname*{i}\omega A\mathbb{E}%
\left[  \mathbf{u}^{\ast}\mathbf{Q}^{3}\mathbf{u}\right]  \mathbb{E}\left[
\mathbf{u}^{\ast}\mathbf{Q}\frac{\mathbf{XT\tilde{E}X}^{\ast}}{N}%
\mathbf{Eu}\right]  \mathbb{E}\left[  \Psi\left(  \omega\right)  \right] \\
+\operatorname*{i}\omega B\mathbb{E}\left[  \mathbf{u}^{\ast}\mathbf{Q}%
^{4}\mathbf{u}\right]  \mathbb{E}\left[  \mathbf{u}^{\ast}\mathbf{Q}%
\frac{\mathbf{XT\tilde{E}X}^{\ast}}{N}\mathbf{Eu}\right]  \mathbb{E}\left[
\Psi\left(  \omega\right)  \right] \\
+\operatorname*{i}\omega B\mathbb{E}\left[  \mathbf{u}^{\ast}\mathbf{Q}%
^{3}\mathbf{u}\right]  \mathbb{E}\left[  \mathbf{u}^{\ast}\mathbf{Q}^{2}%
\frac{\mathbf{XT\tilde{E}X}^{\ast}}{N}\mathbf{Eu}\right]  \mathbb{E}\left[
\Psi\left(  \omega\right)  \right] \\
+\operatorname*{i}\omega A\frac{1}{1-\gamma\tilde{\gamma}}\mathbb{E}\left[
\mathbf{u}^{\ast}\mathbf{Q}^{2}\mathbf{u}\right]  \mathbb{E}\left[
\mathbf{u}^{\ast}\mathbf{Q}\frac{\mathbf{X}\mathbf{T\tilde{E}}\mathbf{X}%
^{\ast}}{N}\mathbf{E}^{2}\mathbf{u}\right]  \mathbb{E}\left[  \Psi\left(
\omega\right)  \right] \\
-\operatorname*{i}\omega A\frac{\gamma}{1-\gamma\tilde{\gamma}}\mathbb{E}%
\left[  \mathbf{u}^{\ast}\mathbf{Q}^{2}\mathbf{u}\right]  \mathbb{E}\left[
\mathbf{u}^{\ast}\mathbf{Q}\frac{\mathbf{X\mathbf{T\tilde{E}}}^{2}%
\mathbf{X}^{\ast}}{N}\mathbf{Eu}\right]  \mathbb{E}\left[  \Psi\left(
\omega\right)  \right] \\
+\operatorname*{i}\omega B\frac{1}{1-\gamma\tilde{\gamma}}\mathbb{E}\left[
\mathbf{u}^{\ast}\mathbf{Q}^{3}\mathbf{u}\right]  \left[  \mathbf{u}^{\ast
}\mathbf{Q}\frac{\mathbf{X}\mathbf{T\tilde{E}}\mathbf{X}^{\ast}}{N}%
\mathbf{E}^{2}\mathbf{u}\right]  \mathbb{E}\left[  \Psi\left(  \omega\right)
\right] \\
-\operatorname*{i}\omega B\frac{\gamma}{1-\gamma\tilde{\gamma}}\mathbb{E}%
\left[  \mathbf{u}^{\ast}\mathbf{Q}^{3}\mathbf{u}\right]  \mathbb{E}\left[
\mathbf{u}^{\ast}\mathbf{Q}\frac{\mathbf{X\mathbf{T\tilde{E}}}^{2}%
\mathbf{X}^{\ast}}{N}\mathbf{Eu}\right]  \mathbb{E}\left[  \Psi\left(
\omega\right)  \right] \\
+\operatorname*{i}\omega B\frac{1}{1-\gamma\tilde{\gamma}}\mathbb{E}\left[
\mathbf{u}^{\ast}\mathbf{Q}^{2}\mathbf{u}\right]  \mathbb{E}\left[
\mathbf{u}^{\ast}\mathbf{Q}^{2}\frac{\mathbf{X}\mathbf{T\tilde{E}}%
\mathbf{X}^{\ast}}{N}\mathbf{E}^{2}\mathbf{u}\right]  \mathbb{E}\left[
\Psi\left(  \omega\right)  \right] \\
-\operatorname*{i}\omega B\frac{\gamma}{1-\gamma\tilde{\gamma}}\mathbb{E}%
\left[  \mathbf{u}^{\ast}\mathbf{Q}^{2}\mathbf{u}\right]  \mathbb{E}\left[
\mathbf{u}^{\ast}\mathbf{Q}^{2}\frac{\mathbf{X\mathbf{T\tilde{E}}}%
^{2}\mathbf{X}^{\ast}}{N}\mathbf{Eu}\right]  \mathbb{E}\left[  \Psi\left(
\omega\right)  \right]  +\mathcal{O}(N^{-1/2})\text{,}%
\end{multline*}
and the result follows readily by using Propositions \ref{theoremEV1} to
\ref{theoremEV4}.

\section{Proof of Proposition \ref{proposition_sigmaAB}%
\label{AppProposition_sigmaAB}%
}

We will only prove the case $B_{M}\neq0$, such that $\inf_{M\geq1}%
B_{M}=B_{\inf}>0$ (the complementary situation is much easier to handle).
Consider first writing%
\[
\sigma_{\xi,M}^{2}\left(  A_{M},B_{M}\right)  =\left(  \frac{B_{M}%
\mathbf{u}_{M}^{\ast}\mathbf{E}_{M}^{2}\mathbf{u}_{M}}{\left(  1-\gamma
\tilde{\gamma}\right)  ^{2}}\right)  ^{2}\mathcal{V}_{M}\left(  A_{M}%
,B_{M}\right)
\]
where we have defined%
\begin{multline*}
\mathcal{V}_{M}\left(  A_{M},B_{M}\right)  =\left(  \tilde{\gamma}^{2}\frac
{1}{N}\operatorname*{tr}\left[  \mathbf{E}_{M}^{4}\right]  +\gamma^{2}\frac
{1}{N}\operatorname*{tr}\left[  \mathbf{\tilde{E}}_{N}^{4}\right]  \right) \\
+4\tilde{\gamma}\left(  1-\gamma\tilde{\gamma}\right)  \mathcal{S}_{M}\left(
A_{M},B_{M}\right)  +4\left\{  \tilde{\gamma}^{2}\frac{1}{N}\operatorname*{tr}%
\left[  \mathbf{E}_{M}^{3}\right]  -\gamma\frac{1}{N}\operatorname*{tr}\left[
\mathbf{\tilde{E}}_{N}^{3}\right]  \right\}  \mathcal{T}_{M}\left(
A_{M},B_{M}\right) \\
+\frac{2}{1-\gamma\tilde{\gamma}}\left\{  \tilde{\gamma}^{3}\left(  \frac
{1}{N}\operatorname*{tr}\left[  \mathbf{E}_{M}^{3}\right]  \right)
^{2}-2\gamma\tilde{\gamma}\frac{1}{N}\operatorname*{tr}\left[  \mathbf{E}%
_{M}^{3}\right]  \frac{1}{N}\operatorname*{tr}\left[  \mathbf{\tilde{E}}%
_{N}^{3}\right]  +\gamma^{3}\left(  \frac{1}{N}\operatorname*{tr}\left[
\mathbf{\tilde{E}}_{N}^{3}\right]  \right)  ^{2}\right\}
\end{multline*}
with%
\begin{multline*}
\mathcal{S}_{M}\left(  A_{M},B_{M}\right)  =\frac{1}{4}\left(  \left(
1-\gamma\tilde{\gamma}\right)  \frac{A_{M}}{B_{M}}+2\frac{\mathbf{u}_{M}%
^{\ast}\mathbf{E}_{M}^{3}\mathbf{u}_{M}}{\mathbf{u}_{M}^{\ast}\mathbf{E}%
_{M}^{2}\mathbf{u}_{M}}\right)  ^{2}+\\
+\frac{1}{2}\left[  \frac{\mathbf{u}_{M}^{\ast}\mathbf{E}_{M}^{4}%
\mathbf{u}_{M}}{\mathbf{u}_{M}^{\ast}\mathbf{E}_{M}^{2}\mathbf{u}_{M}}-\left(
\frac{\mathbf{u}_{M}^{\ast}\mathbf{E}_{M}^{3}\mathbf{u}_{M}}{\mathbf{u}%
_{M}^{\ast}\mathbf{E}_{M}^{2}\mathbf{u}_{M}}\right)  ^{2}\right]
\end{multline*}%
\[
\mathcal{T}_{M}\left(  A_{M},B_{M}\right)  =\frac{1}{2}\left(  \left(
1-\gamma\tilde{\gamma}\right)  \frac{A_{M}}{B_{M}}+2\frac{\mathbf{u}_{M}%
^{\ast}\mathbf{E}_{M}^{3}\mathbf{u}_{M}}{\mathbf{u}_{M}^{\ast}\mathbf{E}%
_{M}^{2}\mathbf{u}_{M}}\right)  \text{.}%
\]
By inequality (\ref{supE}) along with Lemma \ref{lemmaBounds1-Gammas} and the
upper and lower bounds of $B_{M}$, we have that%
\\
$\left(  B_{M}\mathbf{u}_{M}^{\ast}\mathbf{E}_{M}^{2}\mathbf{u}_{M}/\left(
1-\gamma\tilde{\gamma}\right)  ^{2}\right)  ^{2}$ is bounded uniformly above
and away from zero.

We now show that%

\begin{equation}
0<\inf_{M\geq1}\mathcal{V}_{M}\left(  A_{M},B_{M}\right)  \leq\sup_{M\geq
1}\mathcal{V}_{M}\left(  A_{M},B_{M}\right)  <+\infty\text{.}\label{ULbounds}%
\end{equation}
Indeed, the upper bound in (\ref{ULbounds}) follows readily by the triangular
inequality and Lemma \ref{lemmaBounds1-Gammas} along with inequalities
(\ref{upperBoundGammas}) and (\ref{supE}) - (\ref{infE}), together with the
uniform upper and lower bounds $A_{\sup}$ and $B_{\inf}$ of, respectively,
$A_{M}$ and $B_{M}$.

In order to prove the lower bound, we first show that $\mathcal{S}_{M}%
\geq\mathcal{T}_{M}^{2}$. Indeed, observe that we can write%
\begin{multline*}
\mathcal{S}_{M}-\mathcal{T}_{M}^{2}=\frac{1}{2}\left[  \frac{\mathbf{u}%
_{M}^{\ast}\mathbf{E}_{M}^{4}\mathbf{u}_{M}}{\mathbf{u}_{M}^{\ast}%
\mathbf{E}_{M}^{2}\mathbf{u}_{M}}-\left(  \frac{\mathbf{u}_{M}^{\ast
}\mathbf{E}_{M}^{3}\mathbf{u}_{M}}{\mathbf{u}_{M}^{\ast}\mathbf{E}_{M}%
^{2}\mathbf{u}_{M}}\right)  ^{2}\right]  =\\
=\frac{1}{2}\frac{\mathbf{u}_{M}^{\ast}\mathbf{E}_{M}^{4}\mathbf{u}%
_{M}\mathbf{u}_{M}^{\ast}\mathbf{E}_{M}^{2}\mathbf{u}_{M}-\left(
\mathbf{u}_{M}^{\ast}\mathbf{E}_{M}^{3}\mathbf{u}_{M}\right)  ^{2}}{\left(
\mathbf{u}_{M}^{\ast}\mathbf{E}_{M}^{2}\mathbf{u}_{M}\right)  ^{2}}%
\geq0\text{,}%
\end{multline*}
where the last statement follows by the Cauchy-Schwarz inequality. This shows
that, by completing the squares%
\begin{multline*}
4\tilde{\gamma}_{M}\left(  1-\gamma_{M}\tilde{\gamma}_{M}\right)
\mathcal{S}_{M}+4\left(  \tilde{\gamma}_{M}^{2}\frac{1}{N}\operatorname*{tr}%
\left[  \mathbf{E}_{M}^{3}\right]  -\gamma_{M}\frac{1}{N}\operatorname*{tr}%
\left[  \mathbf{\tilde{E}}_{M}^{3}\right]  \right)  \mathcal{T}_{M}\geq\\
\geq4\tilde{\gamma}_{M}\left(  1-\gamma_{M}\tilde{\gamma}_{M}\right)  \left(
\mathcal{T}_{M}+\frac{\tilde{\gamma}_{M}^{2}\frac{1}{N}\operatorname*{tr}%
\left[  \mathbf{E}_{M}^{3}\right]  -\gamma_{M}\frac{1}{N}\operatorname*{tr}%
\left[  \mathbf{\tilde{E}}_{M}^{3}\right]  }{2\tilde{\gamma}_{M}\left(
1-\gamma_{M}\tilde{\gamma}_{M}\right)  }\right)  ^{2}\\
-\frac{\left(  \tilde{\gamma}_{M}^{2}\frac{1}{N}\operatorname*{tr}\left[
\mathbf{E}_{M}^{3}\right]  -\gamma_{M}\frac{1}{N}\operatorname*{tr}\left[
\mathbf{\tilde{E}}_{M}^{3}\right]  \right)  ^{2}}{\tilde{\gamma}_{M}\left(
1-\gamma_{M}\tilde{\gamma}_{M}\right)  }\text{.}%
\end{multline*}
Using this in the expression of $\mathcal{V}_{M}\left(  A_{M},B_{M}\right)  $
and grouping terms, we readily see that%
\begin{multline*}
\mathcal{V}_{M}\left(  A_{M},B_{M}\right)  \geq\frac{\tilde{\gamma}_{M}^{2}%
}{\gamma_{M}}\left(  \frac{1}{N}\operatorname*{tr}\left[  \mathbf{E}_{M}%
^{4}\right]  \frac{1}{N}\operatorname*{tr}\left[  \mathbf{E}_{M}^{2}\right]
-\left(  \frac{1}{N}\operatorname*{tr}\left[  \mathbf{E}_{M}^{3}\right]
\right)  ^{2}\right)  +\\
+\frac{\gamma_{M}^{2}}{\tilde{\gamma}_{M}}\left(  \frac{1}{N}%
\operatorname*{tr}\left[  \mathbf{\tilde{E}}_{M}^{4}\right]  \frac{1}%
{N}\operatorname*{tr}\left[  \mathbf{\tilde{E}}_{M}^{2}\right]  -\left(
\frac{1}{N}\operatorname*{tr}\left[  \mathbf{\tilde{E}}_{M}^{3}\right]
\right)  ^{2}\right)  +\\
+4\tilde{\gamma}_{M}\left(  1-\gamma_{M}\tilde{\gamma}_{M}\right)  \left(
\mathcal{T}_{M}+\frac{\tilde{\gamma}_{M}^{2}\frac{1}{N}\operatorname*{tr}%
\left[  \mathbf{E}_{M}^{3}\right]  -\gamma_{M}\frac{1}{N}\operatorname*{tr}%
\left[  \mathbf{\tilde{E}}_{M}^{3}\right]  }{2\tilde{\gamma}_{M}\left(
1-\gamma_{M}\tilde{\gamma}_{M}\right)  }\right)  ^{2}\\
+\frac{1}{\left(  1-\gamma_{M}\tilde{\gamma}_{M}\right)  \gamma_{M}}\left(
\tilde{\gamma}_{M}\frac{1}{N}\operatorname*{tr}\left[  \mathbf{E}_{M}%
^{3}\right]  -\gamma_{M}^{2}\frac{1}{N}\operatorname*{tr}\left[
\mathbf{\tilde{E}}_{M}^{3}\right]  \right)  ^{2}\text{.}%
\end{multline*}
The first two terms are positive due to the fact that%
\[
\frac{1}{N}\operatorname*{tr}\left[  \mathbf{\tilde{E}}_{M}^{4}\right]
\frac{1}{N}\operatorname*{tr}\left[  \mathbf{\tilde{E}}_{M}^{2}\right]
-\left(  \frac{1}{N}\operatorname*{tr}\left[  \mathbf{\tilde{E}}_{M}%
^{3}\right]  \right)  ^{2}\geq0\text{,}%
\]
which is a consequence of the Cauchy-Schwarz inequality (and equivalently for
$\mathbf{E}_{M}$ instead of $\mathbf{\tilde{E}}_{M}$), and this leads to
\[
\mathcal{V}_{M}\left(  A_{M},B_{M}\right)  \geq\frac{1}{\left(  1-\gamma
_{M}\tilde{\gamma}_{M}\right)  \gamma_{M}}\left(  \tilde{\gamma}_{M}\frac
{1}{N}\operatorname*{tr}\left[  \mathbf{E}_{M}^{3}\right]  -\gamma_{M}%
^{2}\frac{1}{N}\operatorname*{tr}\left[  \mathbf{\tilde{E}}_{M}^{3}\right]
\right)  ^{2}\text{.}%
\]
Now, using again the Cauchy-Schwarz inequality we are able to write
\[
\gamma_{M}^{2}=\left(  \frac{1}{N}\operatorname*{tr}\left[  \mathbf{E}_{M}%
^{2}\right]  \right)  ^{2}\leq\frac{1}{N}\operatorname*{tr}\left[
\mathbf{E}_{M}\right]  \frac{1}{N}\operatorname*{tr}\left[  \mathbf{E}_{M}%
^{3}\right]  =\delta_{M}\frac{1}{N}\operatorname*{tr}\left[  \mathbf{E}%
_{M}^{3}\right]  \text{,}%
\]
and this implies that
\begin{multline*}
\tilde{\gamma}_{M}\frac{1}{N}\operatorname*{tr}\left[  \mathbf{E}_{M}%
^{3}\right]  -\gamma_{M}^{2}\frac{1}{N}\operatorname*{tr}\left[
\mathbf{\tilde{E}}_{M}^{3}\right]  \geq\frac{1}{N}\operatorname*{tr}\left[
\mathbf{E}_{M}^{3}\right]  \left(  \tilde{\gamma}_{M}-\delta_{M}\frac{1}%
{N}\operatorname*{tr}\left[  \mathbf{\tilde{E}}_{M}^{3}\right]  \right)  =\\
=\frac{1}{N}\operatorname*{tr}\left[  \mathbf{E}_{M}^{3}\right]  \left(
\frac{1}{N}\operatorname*{tr}\left[  \mathbf{T}_{N}^{2}\left(  \mathbf{I}%
_{N}+\delta_{M}\mathbf{T}_{N}\right)  ^{-3}\right]  \right)  \text{.}%
\end{multline*}
Therefore, we have shown that
\[
\mathcal{V}_{M}\left(  A_{M},B_{M}\right)  \geq\frac{\left(  \frac{1}%
{N}\operatorname*{tr}\left[  \mathbf{E}_{M}^{3}\right]  \left(  \frac{1}%
{N}\operatorname*{tr}\left[  \mathbf{T}_{N}^{-1}\mathbf{\tilde{E}}_{N}%
^{3}\right]  \right)  \right)  ^{2}}{\left(  1-\gamma_{M}\tilde{\gamma}%
_{M}\right)  \gamma_{M}}\text{,}%
\]
and the lower bound in (\ref{ULbounds}) finally follows from Lemma
\ref{lemmaBounds1-Gammas} together with inequalities (\ref{upperBoundGammas}),
(\ref{infE}) and (\ref{infEtilde}).

\bibliographystyle{IEEEtran}
\bibliography{ReferencesCLT}
%

\newpage
%

\begin{figure}
\centerline{\includegraphics[width=\linewidth]{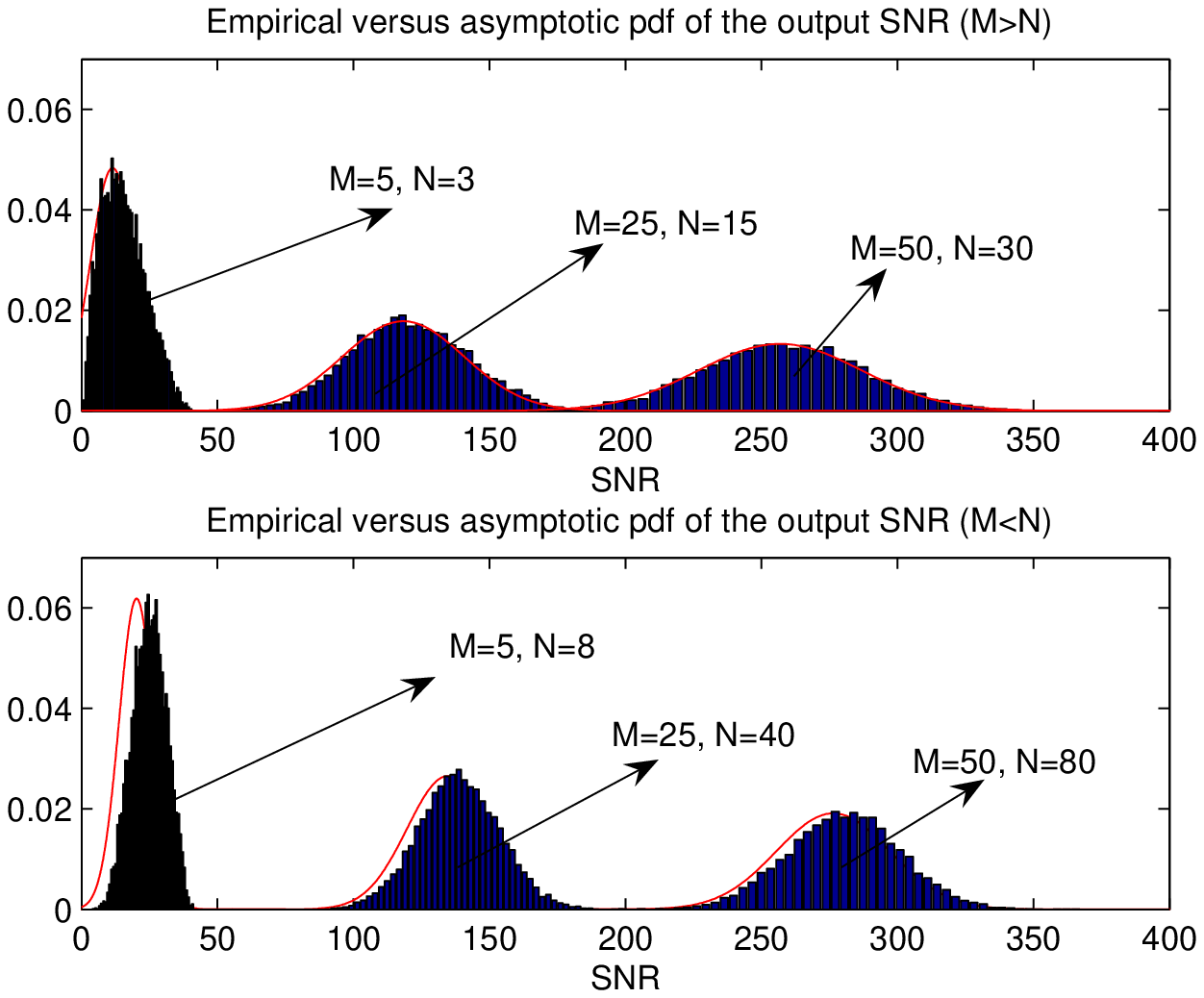}}
\caption
{Numerical evaluation of fitness accuracy of CLT (Supervised Training).}
\label{figure1}
\end{figure}
\begin{figure}
\centerline{\includegraphics[width=\linewidth]{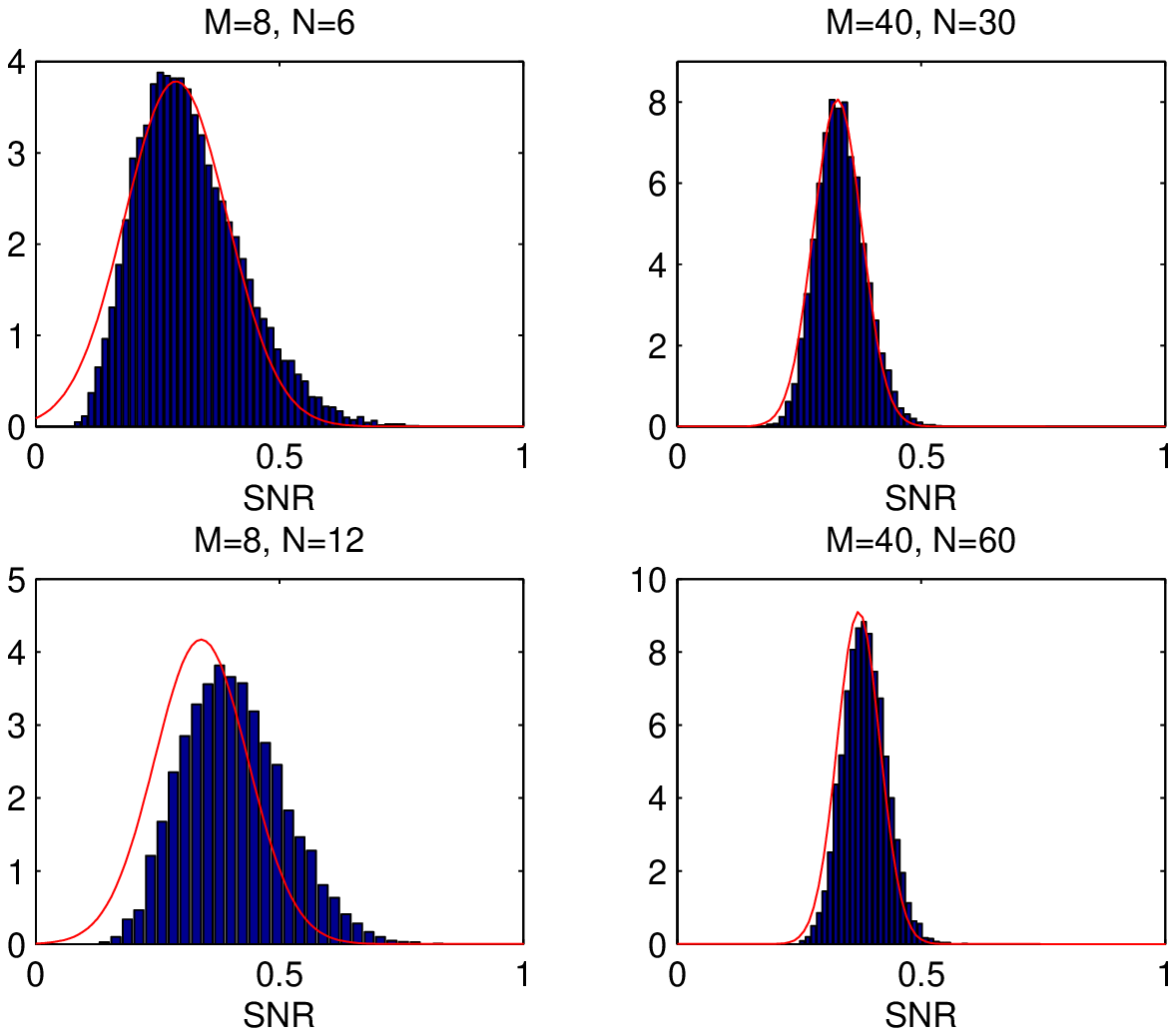}}
\caption
{Numerical evaluation of fitness accuracy of CLT (Unsupervised Training).}
\label{figure2}
\end{figure}%

\end{document}